\documentclass[12pt]{article}
\usepackage{amsfonts}
\usepackage{eurosym}
\usepackage{graphicx, tikz, tikzscale}
\usepackage{booktabs}
\usepackage{float}
\usepackage{subfig}
\usepackage{amssymb}
\usepackage{comment}
\usepackage{amsmath, bbm, amsthm}
\usepackage{booktabs}
\usepackage{multirow}
\usepackage{enumerate,multicol,enumitem}
\usepackage{color}
\usepackage[authoryear]{natbib}
\usepackage{hyperref}
\usepackage{setspace}
\usepackage[title]{appendix}

\hypersetup{colorlinks=true, linkcolor=red, citecolor=blue}

\newtheorem{claim}{{\bf \sc Claim}}

\newtheorem{corollary}{{\bf \sc Corollary}}
\newtheorem{proposition}{{\bf \sc Proposition}}
\newtheorem{remark}{{\bf \sc Remark}}
\theoremstyle{definition}
\newtheorem{example}{{\sc Example}}

\def\eproof{\hbox{\hskip3pt\vrule width4pt height8pt depth1.5pt}}

\begin{document}
\title{
Optimal Regulation and Investment Incentives in
Financial Networks.\thanks{This paper contains results that were originally in ``Distorted Investment Incentives, Regulation, and Equilibrium Multiplicity in a Model of Financial Networks.''
We split that paper into this one, and another (``Credit Freezes, Equilibrium Multiplicity, and Optimal Bailouts in Financial Networks'' {\sl Review of Financial Studies}, 2024).
We have added a number of new results to both papers beyond the results on regulation that appeared in the original paper.   }
}
\author{Matthew O. Jackson and Agathe Pernoud\thanks{%
Matthew O. Jackson; Department of
Economics, Stanford University, Stanford, California 94305-6072 USA; jacksonm@stanford.edu.
Jackson is also an
external faculty member of the Santa Fe Institute. Agathe Pernoud; Booth School of Business, University of Chicago, Chicago, Illinois 60637 USA; agathe.pernoud@chicagobooth.edu.  We gratefully
acknowledge financial support under NSF grant SES-1629446.
We thank Marco Bardoscia, Celso Brunetti, Ozan Candogan, Matt Elliott, Gerardo Ferarra,
David Hirshleifer, Arjun Mahalingam, Carlos Ramirez, Tarik Roukny, Alireza Tahbaz-Salehi, Erol Selman, and especially
Co-Pierre Georg, for
conversations and comments, and the editor and referees for helpful suggestions.
}}
\date{Draft: October 2025 
}
\maketitle

\begin{abstract}
We examine optimal regulation of financial networks with debt interdependencies between financial firms.
We first show that firms often have an incentive
to choose excessively risky portfolios
and overly correlate their portfolios with those of their counterparties.
We then characterize how optimal regulation depends on a firm's financial
centrality and its available investment opportunities.
In standard core-periphery networks, optimal regulation depends non-monotonically on the correlation of banks' investments, with maximal restrictions for intermediate levels of correlation.
Moreover, it can be uniquely optimal to treat banks asymmetrically: restricting the investments of one core bank
while allowing an otherwise identical core bank (in all aspects, including network centrality) to invest freely.


\end{abstract}

\setcounter{page}{0}\thispagestyle{empty} \newpage

\section{Introduction}

In financial networks with missing markets, in which not all risks can be fully hedged, defaults and
counterparty risk impose externalities that can be large sources of inefficiency. Defaults involve substantial deadweight costs that can compound if left to cascade. The potential for disaster is exacerbated if many organizations are invested in similarly
distressed portfolios, as was the case in 2008.
These concerns have led governments to bail out financial institutions, often only after the risk of
systemic failure had already reached a critical and expensive point.

The costs of such interventions can be substantially lower if appropriate and well-targeted
regulation prevents financial institutions
from undertaking investments that generate excessive systemic risk.
Optimal regulation, however, involves a tradeoff between reducing risk and foregoing valuable investments,
and thus requires understanding the incentives financial institutions have in choosing their investments.
Although there is a growing literature on financial networks and their
consequences,\footnote{For a recent survey, see Jackson and Pernoud \citeyearpar{jacksonp2020survey}.}
there is no systematic study of optimal regulation that accounts for financial institutions' incentives, and especially as a function of their network positions.

In this paper, we examine how the network structure distorts financial institutions' investment incentives
and derive the associated optimal regulatory response.
Our analysis focuses on bankruptcies and
the discontinuous costs upon insolvency, since those play a key role
in inefficiencies and externalities in financial networks.
Bankruptcy costs impose deadweight losses borne by some agent(s), and hence
reduce total welfare regardless of whether the defaulting party faces limited liability.
Nonetheless, the analysis also applies to credit freezes and the
resulting deadweight losses from the lost liquidity.

Our analysis provides three contributions.

First, we examine banks' incentives when making investment decisions, and characterize
how network externalities result in inefficiencies.
One result shows that, under general conditions on the network structure,
banks choose to take on excess risk compared to what is socially efficient. This comes
from the fact that banks do not account for the negative externalities their defaults impose
on the overall financial system when choosing their investments. In addition, we show that banks have strong incentives to excessively correlate their investments with those of
their counterparties.
This happens for a reason that we refer to as `risk matching.'
The intuition is that a bank wishes to be solvent when it receives the most payments from other banks, and wishes to be insolvent when
other banks are defaulting.
This induces complementarities in investment incentives and results in matched investments across banks, thus bringing a new aspect to risk shifting. The above intuition holds regardless of whether there is limited liability, so it is more involved than standard examples of risk shifting. Nonetheless, it follows a simple principle: a bank wants to be solvent when others are solvent so as to
benefit from that interaction.

Second, we examine optimal regulation by a government
that maximizes total expected returns to investments net of bankruptcy costs.
We analyze three regulatory approaches: macroprudential regulation, bailouts,
and laissez-faire.  We start by considering networks that admit a core-periphery structure, and we examine how optimal regulation depends on the correlation between banks' investment opportunities.
Here the results are particularly nuanced.  We show that placing restrictions on banks' investments is optimal for intermediate levels
of correlation, but not for sufficiently high or low correlation.
The reasons for the absence of regulation at the extremes differ.
With low levels of correlation,
any single bank's default is unlikely to affect others since they are likely to have high returns and so all be resilient.   With high levels of correlation,
any single bank's default is unlikely to affect others since they are already likely to be defaulting as well.  So, with extreme levels of correlation
there is an absence of contagion, but for different reasons.

A related result is that, in the intermediate range of correlation in which restrictions improve over laissez-faire,
the optimal restrictions are often asymmetric.  In particular, imposing restrictions on some core banks
but not on others that are completely identical in every way to the regulated ones can provide strictly higher social welfare than any regulation that
treats identical banks symmetrically.
The intuition is roughly as follows.  Core banks can generally withstand some base number of defaults, regardless of the level of regulation.
Allowing some number of core banks (below that base threshold number) to invest freely,
allows society to earn high returns from risky investments, while still avoiding contagion.
Instead, a symmetric regulation across banks is either too restrictive and
eliminates
contagion but also misses out on high returns, or too lax and admits widespread
contagion.  An asymmetric regulation enables some high investment returns, but also avoids widespread contagion.
In practice, asymmetric regulation might be difficult legally, but can be done to the extent
that there are even slight asymmetries among banks in the core -- for instance regulating ones
above a particular size.

Finally, we provide more general insights on how the network structure shapes optimal regulation, beyond core-periphery networks. The network structure enters the calculations because it determines
the consequences of any given bank's default on
the rest of the financial system.
Optimal regulation depends on a particular definition of centrality, which we develop and that summarizes a bank's marginal contribution to overall bankruptcy costs. A regulator trades off higher expected returns
from risky portfolios against systemic risks.
The characterization of the optimal policy then favors macroprudential regulation (in the form of a limit on portfolio risk)
if the loss in expected returns
is small relative to the bank's centrality.
If expected investment returns are high enough, then the optimal policy is not to impose any limits on investments
but to intervene via bailouts for sufficiently central organizations and to allow less central organizations to fail if they become insolvent.
These regions of optimal regulation are previewed in Figure \ref{figure-regulate}, and we leave precise definitions of variables to Section \ref{reg}.
\begin{figure}[h!]
\centering
\includegraphics[width=0.4\textwidth]{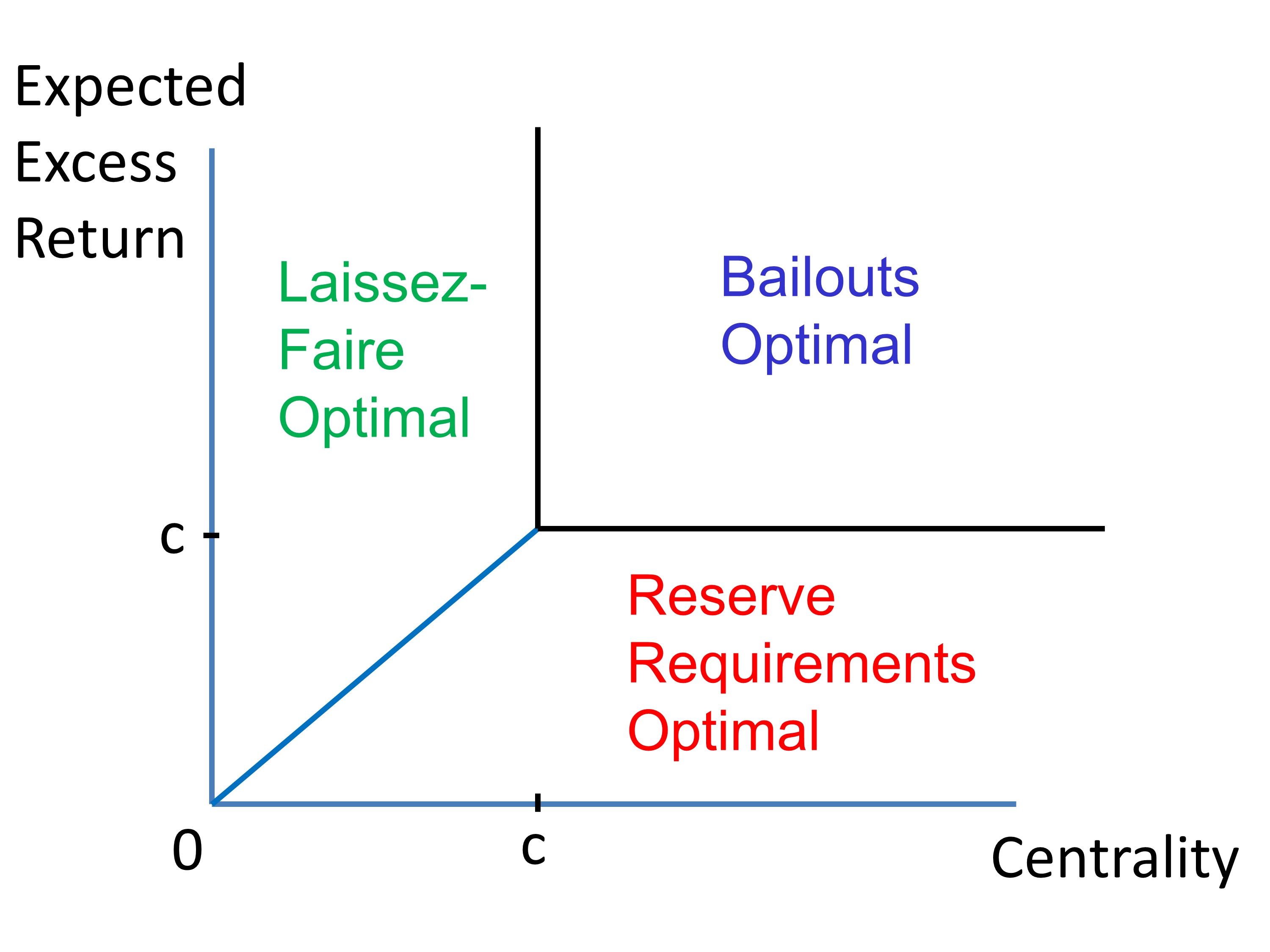}
\caption{\label{figure-regulate}
The optimal regulation of a bank as a function of the expected excess return of the bank's available risky investments and the centrality of the bank.}
\end{figure}

\paragraph{Related Literature:}

Several papers study how connections between financial firms distort their investment incentives, but none has looked at optimal regulation in a network setting.\footnote{
There is a literature on the
more standard agency problem between the manager of a firm and its shareholders,
for example, as highlighted in Jensen and Meckling \citeyearpar{jensenm1976} and Galai and Masulis \citeyearpar{galai1976option},
which can induce inefficient investments.
For more on this point, see Admati and Hellwig \citeyearpar{admatih2013}.  
For a more general discussion of agency problems of excessive risk-taking in the presence of externalities see
Hirshleifer and Teoh \citeyearpar{hirshleifert2009}.}

With regards to part of our analysis on incentives, an (independent) study by Shu \citeyearpar{shu2019} also finds that network externalities between banks lead them to take more risk and correlate their investments. His model is complementary to ours as it differs from ours in key ways.
Importantly, our analysis of optimal regulation depends on features not in his model: the centrality of
particular financial institutions and asymmetries across them. Zawadowski \citeyearpar{zawadowski2013entangled} considers banks on a ring network that can use OTC contracts with their neighbors to hedge investment risks. He shows that banks fail to insure against counterparty risk, hence allowing for inefficiently large contagion risks. Galeotti and Ghiglino \citeyearpar{galeotti2021cross} study the portfolio choices of firms that are embedded in a network of cross-holdings and have mean-variance preferences. They show that cross-holdings provide implicit insurance to firms, and that the lower a firm's self-ownership, the riskier its investment.
Vohra, Xing,  and Zhu \citeyearpar{vohra2020network} consider a network of firms that can make an
investment \emph{following} a shock to attenuate it, and study how agency conflicts within firms
impact contagion.\footnote{Both  Galeotti and Ghiglino \citeyearpar{galeotti2021cross}
and Vohra, Xing,  and Zhu \citeyearpar{vohra2020network}
restrict attention to equity-like interdependencies between firms.}

There are also papers highlighting banks' incentives to herd. Elliott, Georg, and Hazell  \citeyearpar{elliottgh2018} examine which network of inter-bank equity claims and correlation structure of investments arise endogenously in equilibrium when banks act
under limited liability. Their main result is that banks have an incentive to partner
with counterparties whose
portfolios are positively correlated with their own in order to shift losses from shareholders
to creditors.  Using data on the German banking system,
they also provide strong empirical evidence that banks lend more to those with portfolios
similar to their own. Thus, such correlations are observed. A few studies put forward other channels that can lead banks to want to correlate their portfolios. Acharya \citeyearpar{acharya2009theory} considers a setting in which, when a bank defaults and fails to repay depositors, the aggregate supply of funds in the economy shrinks, which hurts surviving banks. He argues that, if banks act under limited liability, then this negative externality can incentivize banks to herd when choosing their investments. Capponi and Weber \citeyearpar{capponi2022systemic} study banks' choice of portfolio in the presence of fire-sale spillovers. They show that banks choose portfolios that are overly similar, leading to under-diversification at the systemic level.
Arya and Golver \citeyearpar{arya2001} and Acharya and Yorulmazer \citeyearpar{acharya2007} highlight how the possibility of being bailed-out
can lead agents to herd in order to capture the bailout subsidies.\footnote{Acharya and Yorulmazer \citeyearpar{acharya2008cash} consider an alternative policy, which consists in granting liquidity to surviving banks so as to subsidize purchases of liquidated assets. They show this can be equivalent to a bailout policy ex post while strictly dominating it ex ante, as it incentivizes banks to differentiate their investments as opposed to correlating them.} Similarly, Farhi and Tirole \citeyearpar{farhi2012collective} show that banks' choice of leverage can become strategic complements given a regulator's optimal monetary policy response.
Note that in these last four papers, the inefficiency does not stem from network externalities as banks are not counterparties of each other.
Thus, correlation incentives arise in our model
for partly different reasons, and persist even if banks account for missed payments to debtholders. We show that, because of
financial interdependencies, bank values depend positively on each other, which induces complementarities
in their returns to investments: a high return for a bank is weakly more valuable if its partners have
high returns as well, pushing them to correlate their portfolios.

\section{A Model of Financial Interdependencies}

\subsection{The Interbank Network}

Consider a set $N= \{0, 1,\ldots,n\}$ of institutions.  We treat $ \{1,\ldots, n\}$ as the financial institutions and refer to them as ``banks" for simplicity in terminology.
Node $0$ encompasses all other actors, who either hold debt in the banks (for instance, private investors and depositors) or have debts to the banks (for instance, private and public companies), but not both.
We generally take node 0's connections as given, and our analysis of the model centers on the banks $i>0$.

Banks' portfolios are composed of investments outside of the network as well as financial contracts within the network (i.e., to other banks).
Abusing notation, let $K=\{1,\ldots K\}$ be the set of primitive investment opportunities, or assets,
whose prices are normalized to one.\footnote{This is without loss of generality, as these can be thought of as how much a given dollar buys, and then we can capture relative returns via the distribution of realized values.}
We denote by $p_k$ the gross return or realized value of asset $k\in K$.  Let $q_{ik}\geq 0$ be the quantity invested in asset $k$ by Bank $i$, and $\mathbf{q}$ the matrix whose $(i,k)$-th entry is equal to $q_{ik}$. Analogous notation is used for all matrices and vectors.
Throughout, when we refer to ``returns'' we are referring to ``gross returns''.

For the purposes of defining defaults and bankruptcies, we take realized gross returns $\mathbf{p}$ and portfolios $\mathbf{q}$ as given. In Section \ref{model:investments} we describe how banks choose their portfolios.

In addition to these direct outside investments, banks have
claims on and liabilities to other banks in the network.
These interdependencies take the form of bilateral debt contracts where $D_{ij}\geq 0$ denotes the face value of $j$'s liability towards $i$, that is the promised payment between the debtor $j$ and its creditor $i$.  A bank's total nominal debt assets and liabilities are, respectively, $D_i^{A}\equiv \sum_j D_{ij} $
and $D_i^{L}\equiv \sum_j D_{ji}$.
The matrix of debt claims $\mathbf{D}\equiv [D_{ij}]_{ij}$ defines the weighted directed interbank network. Throughout the paper, we take the network of debts as given and common knowledge.

In an online appendix we show how the model extends to allow for more general contracts beyond debt, including both debt and equity contracts between banks. The results are slightly more subtle and require an additional condition on the network structure regarding cycles of cross ownership.  Even though such a model shows the robustness of the results, the intuitions are easiest to see when defaults are driven by debts alone.

\subsection{Bankruptcies and Equilibrium Values of Banks} \label{sec:values}

Financial contracts between banks create interdependencies between their balance sheets.

Let $\mathbf{V}=(V_1,\ldots V_n)$ denote the vector of bank values, which summarize the values of banks' assets (investments and debt repayments coming in) net of debt liabilities and any bankruptcy costs. In this section, we describe how to compute these values, given banks' portfolios $ \mathbf{q}$, realized asset returns $\mathbf{p}$, and interbank debts $\mathbf{D}$.

A bank is said to default if the value of
its assets is not enough to cover all of its liabilities, in which case it incurs bankruptcy
costs $\beta_i(\mathbf{V}; \mathbf{q}, \mathbf{p},\mathbf{D})\geq 0$. These costs can stem from legal bankruptcy proceedings, the early liquidation of productive assets, fire sales, or inefficient renegotiations with counterparties.\footnote{Efficient renegotiations of contracts might be precluded due to frictions such as asymmetric information, as in Glode and Opp \citeyearpar{glode2021private}.} 
We allow these
costs to depend on the health of other banks as well as the value of the primitive investments and the debt structure.

Since a bank can default on its liabilities, the amount that it actually repays to its creditors
depends on the value of its balance sheet, and hence ultimately on the full vector of values $\mathbf{V}$.
Denote by $d_{ij}(\mathbf{V};\mathbf{D})$ the realized debt payment from a debtor $j$ to one of its creditors $i$. As will become clear, these can be written as a function solely of the vector of bank values and interbank debts.

A Bank $j$ must repay its debts in full whenever possible.  Hence, if $j$ is solvent ($V_j\geq 0$), then
\[
d_{ij}(\mathbf{V};\mathbf{D}) =   D_{ij}.
\]
If instead $V_j<0$ and $j$ defaults,
then its creditors become the residual claimants on its assets and are rationed proportionally to their claim on $j$:
\begin{equation}
\label{debtvalue}
d_{ij}(\mathbf{V};\mathbf{D}) =   \frac{D_{ij}}{\sum_h D_{hj}} \max\left(V_j+D_j^L,0\right).
\end{equation}

A bank defaults whenever the value of its assets does not cover its liabilities.
The bankruptcy costs it incurs are
\begin{equation}
\label{bankrupt}
  b_i\left(\mathbf{V}; \mathbf{q}, \mathbf{p},\mathbf{D}\right) =
  \begin{cases}
   0 & \text{ if  } \sum_k q_{ik} p_k   +d_i^A(\mathbf{V};\mathbf{D})  \geq D_i^L \\
   \beta_i\left(\mathbf{V}; \mathbf{q}, \mathbf{p},\mathbf{D}\right) & \text{ if  } \sum_k q_{ik} p_k +d_i^A(\mathbf{V};\mathbf{D})  <D_i^L.
  \end{cases}
\end{equation}
where $d_{j}^A(\mathbf{V};\mathbf{D})\equiv \sum_{h}d_{jh}(\mathbf{V};\mathbf{D})$.

We consider bankruptcy costs that have the following general form:
\[\beta_i(\mathbf{V}; \mathbf{q}, \mathbf{p},\mathbf{D})=\chi+a\left[ \sum_k q_{ik} p_k
+d_i^A(\mathbf{V};\mathbf{D})\right]\quad\text{with $a\in[0,1]$, $\chi\geq 0$}.\]
That is, a bank loses a fraction $a$ of its assets upon defaulting, for instance because it only recovers assets partially upon liquidation, and incurs some additional fixed cost $\chi$.\footnote{This functional form for bankruptcy costs embeds most of the costs considered in the literature. }

Bankruptcy costs depress the value of a bank.  The value of a bank is:
$$
V_i=\sum_k q_{ik} p_k  +   d_i^A(\mathbf{V};\mathbf{D})-D^L_i  - b_i(\mathbf{V}; \mathbf{q}, \mathbf{p},\mathbf{D}),
$$
which is Bank $i$'s realized assets net of debts and bankruptcy costs,
where bankruptcy costs $b_i(\mathbf{V}; \mathbf{q}, \mathbf{p},\mathbf{D} )$ are defined by (\ref{bankrupt}) and realized debt payments in are $ d_j^A(\mathbf{V};\mathbf{D})$ by (\ref{debtvalue}).

This produces a fixed point problem, since banks' values appear on both sides of this equation.
Thus we examine banks' equilibrium values, which are the solution to:\footnote{We use the term ``equilibrium values'' to refer to a fixed point of
 equation (\ref{eq-bookvalue-bankruptcy}) to keep with the literature, but note that the
 term ``equilibrium''  is also used below for situations in which we endogenize investment choices.}
\begin{equation}
\label{eq-bookvalue-bankruptcy}
\mathbf{V}= \mathbf{q} \mathbf{p} + \mathbf{d}^A(\mathbf{V};\mathbf{D})  - \mathbf{D}^L - \mathbf{b}(\mathbf{V}; \mathbf{q}, \mathbf{p},\mathbf{D}).
\end{equation}

Note that the value of a bank is weakly increasing in those of others in the network, which ensures the existence of a solution to equation (\ref{eq-bookvalue-bankruptcy}).\footnote{This is implied by Tarski's fixed point theorem, since bank values depend monotonically on each other and are bounded. They are bounded above by $\overline{\textbf{V}} = \textbf{q}\textbf{p}+\textbf{D}^A$.   }
Since defaults generate discontinuities, there can even exist multiple solutions, and hence multiple vectors of equilibrium bank values $\mathbf{V}$ for given portfolios $\mathbf{q}$, that necessarily form a complete lattice.\footnote{For an extensive analysis of the equilibrium structure, see \cite{jacksonp2020,csokah2024}. In \cite{jacksonp2020}, we study in detail when network interdependencies can lead to multiple solutions for bank values and how the regulator can prevent this multiplicity by bailing out some banks. There we take banks' portfolios as given, contrary to the present paper which focuses on the (in-)efficiency of their investment decisions.} In what follows, we focus on the best equilibrium for bank values.
That is, for given portfolios, realized gross returns, and interbank claims $(\mathbf{q}, \mathbf{p}, \mathbf{D})$, the bank values $\mathbf{V}(\mathbf{q}, \mathbf{p}, \mathbf{D})$ are 
given by the greatest solution to equation (\ref{eq-bookvalue-bankruptcy}).%
\footnote{Alternatively, we could focus on the worst equilibrium, or any other selection of equilibrium as long as we fixed the selection criterion throughout the analysis, and the insights below would remain.  If one allows the selection of the equilibrium to change based on banks' behaviors, then that can result in incentive changes.  For instance, if a bank anticipates the best equilibrium for bank values if it chooses one portfolio, but the worst equilibrium if it picks another, then that could induce it to choose the first portfolio even that is dominated  in terms of returns by the second.} Banks' equity values correspond to $\mathbf{V}(\mathbf{q}, \mathbf{p}, \mathbf{D})^+:=\max\{\mathbf{V}(\mathbf{q}, \mathbf{p}, \mathbf{D}),0\}$.

We end this section with a brief discussion of how to interpret bankruptcy costs and who incurs them. As noted above, these costs capture anything that discontinuously depresses a bank's value upon insolvency, such as legal costs or inefficient liquidation of investments. These costs are imposed on a bank's balance sheet, and so are first and foremost incurred by the bank. If they exceed the value of the bank's assets, then the bank's value $V_i$ becomes negative. (However, a bank's equity value $V_i^+$ is always non-negative.) In that case, one can interpret these remaining costs as being absorbed by a government or the public. What is important for our analysis is that these are real economic costs, which reduce total surplus.

\subsection{Banks' Investment Decisions}\label{model:investments}

A key contribution of our model is to endogenize banks' investments in the outside assets $\mathbf{q}$.

Each bank has some capital to invest (which by default is one unit, unless otherwise stated) and chooses how to allocate that among a variety of outside assets. To allow for the possibility that some assets might not be available to all banks, we let $K_i\subset K$ denote the assets that Bank $i$ can invest into and require that $q_{ik} = 0$ for all $k\notin K_i$.  We assume that the decision-makers in a bank (i.e., its shareholders) maximize the bank's expected equilibrium (equity) value. 
This is only restrictive if some shareholders are also significant creditors of the bank (e.g., depositors), in which case their incentives to get their debt payment in full might create market discipline.  The main incentive issues that we examine are then between the shareholders of a bank and the other agents who are impacted by the bank's investment decision through network interdependencies.

Asset returns $\mathbf{p}$ are uncertain when investment decisions are made, and drawn from some common knowledge distribution. We take the bank's shareholders to be risk neutral, which allows us to abstract away from risk differences as the driver of moral hazard and to focus on systemic and structural externalities.\footnote{We show in Section \ref{riskaversion} of the Online Appendix how the results extend to the case of risk aversion. }
Thus, shareholders choose $\mathbf{q_i}$ to maximize $\mathbb{E}_\mathbf{p} \left[V_i(\mathbf{q_i}, \mathbf{q_{-i}},\mathbf{p};\mathbf{D})^+\right]$, that is they solve
\[
\max_{(q_{ik})_{k\in K_i}\in\Delta(K_i)} \, \mathbb{E}_\mathbf{p} \left[ \left(\sum_kq_{ik}p_k
+ \sum_jd_{ij}(\mathbf{V}(\mathbf{q_i}, \mathbf{q_{-i}}, \mathbf{p};\mathbf{D}),\mathbf{D}) -D_i^{L} \right)^+ \right].
\]

The vector of banks' portfolios, $\mathbf{q}^*=(\mathbf{q^*_i},\mathbf{q^*_{-i}})$, form a Nash equilibrium if, for each $i$,  $\mathbf{q^*_i}$ maximizes  Bank $i$'s expected equity value given that other banks choose portfolios $\mathbf{q^*_{-i}}$.

In the second part of the paper, we consider two forms of regulation: \emph{ex-ante} macroprudential regulation that constrains which portfolios ($\mathbf{q}$) banks can choose, and  \emph{ex-post} bailouts.

The ex-ante macroprudential regulation takes the form of some restrictions on portfolio choices, often coming in the form of capital regulation or reserve requirements.  For instance, the regulator may require that a bank's portfolio consist 
of at least a minimum investment in ``risk-free'' securities.

An ex-post bailout of a Bank $i$ means that the bank's value is brought back to the solvency threshold  and all its debt claims are settled: $V_i(\textbf{q}, \textbf{p}; \textbf{D})=0$ and $d_{ji}(\textbf{q}, \textbf{p}; \textbf{D}) = D_{ji}$ for all $j$.

Overall, the timing of the model is as follows:
\begin{figure}[!h]
\begin{center}
\begin{tikzpicture}[scale=0.8]
\draw[->, thick] (-3,0) to (15,0);
\draw[thick] (-2.5,0.2) to (-2.5,-0.2);
\draw (-2.5,-0.7) node {\small \emph{Ex ante}};
\draw (-2.5,-1.2) node {\small \emph{regulation}};
\draw[thick] (0.5,0.2) to (0.5,-0.2);
\draw (0.5,-0.7) node {\small Banks choose};
\draw (0.5,-1.3) node {\small portfolios $\textbf{q}$};
\draw (4,-0.7) node {\small Asset returns };
\draw (4,-1.2) node {\small $\textbf{p}$ are realized};
\draw[thick] (4,0.2) to (4,-0.2);
\draw (8,-0.7) node {\small \emph{Ex post}};
\draw (8,-1.2) node {\small  \emph{Bailouts made}};
\draw[thick] (8,0.2) to (8,-0.2);
\draw[thick] (13,0.2) to (13,-0.2);
\draw (13,-0.7) node {\small Values $ \textbf{V}(\textbf{q}, \textbf{p}; \textbf{D})$ and};
\draw (13,-1.2) node {\small  bankruptcy costs result };
  \end{tikzpicture}
  \end{center}
\end{figure}

\section{A Motivating Example}\label{example}

We start by presenting an example that illustrates the results and
highlights the role of network externalities in generating inefficiencies.

\begin{example}\label{ex:1}
Consider the network depicted in Figure \ref{fig:ex1}. There are two banks: Bank 2 owes an amount $2D$ to Bank 1, while Bank 1 owes $D$ to outside investors. One interpretation is that Bank 2 is a core investment bank that borrowed some capital from a more peripheral Bank 1, which acts as intermediary between depositors and Bank 2.  This is one of the simplest possible examples to illustrate the logic, and clearly can be enriched and embedded in larger networks.
\begin{figure}[!h]
\begin{center}
\begin{tikzpicture}[scale=1]
\definecolor{afblue}{rgb}{0.36, 0.54, 0.66};
\foreach \Point/\PointLabel in {(0,0)/2, (4,0)/1}
\draw[fill=afblue!40] \Point circle (0.35) node {$\PointLabel$};
\draw[->, thick] (0.5,0) to node[above] {$D_{12}=2D$} (3.5,0);
\draw[->, thick] (4.5,0) to node[above] {$D_{01}=D$} (7.5,0);
\draw[fill=black!10] (8,0) circle (0.35) node {$0$};
  \end{tikzpicture}
  \end{center}
  \captionsetup{singlelinecheck=off}
  \caption[]{\small Bank 2 has a debt liability of $2D$ towards Bank 1. Bank 1 has a debt liability of $D$ to outside investors (node 0). }\label{fig:ex1}
\end{figure}

Each bank has a unit of capital to invest either in a risk-free asset with deterministic gross return $1+r$,
or in a risky asset with uncertain return $p$.
The risky asset yields $p=R$ with probability $\theta$ and $p=0$ otherwise,
with $\theta R>1+r>2D$.
For the sake of this example, let $a=0$ and $\chi>0$ such that a bank incurs a fixed cost $\chi$ if it defaults.
Since $1+r>2D$, bankruptcy can be avoided by investing sufficiently in the risk-free asset. Whether a bank defaults when the risky asset does not pay off then depends on its portfolio as well as its counterparty's.

Let $q_i\in[0,1]$ denote the share of capital that Bank $i$ invests in the risky asset, with $1-q_i$ in the risk-free asset (so there are no short sales or leveraging).
Since interbank payments are transfers between agents, society's full payoff as a function of banks' portfolio choices is
\[
\mathbb{E}\left[ \sum_{i=1,2} q_i p + (1-q_i)(1+r)\right] - \chi \sum_{i=1,2}\Pr\left[ V_i(q_1,q_2;p)<0\right].
\]
The socially efficient portfolios balance the expected excess return associated
with investments in the risky asset ($q_i [\mathbb{E}\left[ p \right]- (1+r)]$) against the potential costs of bankruptcy ($ \chi \Pr\left[ V_i(q_1,q_2;p)<0\right]$).

Let us characterize the socially efficient portfolios.  The expected return to the risky portfolio is higher than the risk-free rate,
and so the only reason to invest any amount in the risk-free asset is to avoid bankruptcy costs (given that everyone is risk neutral).
Note that, if Bank 2 remains solvent when $p=0$, then it pays its debt to Bank 1 which then remains solvent as well.
Thus, the most relevant cases to consider are (i) Bank 2 invests enough in the risk-free asset to always be solvent (in which case Bank 1 is always solvent and then should be investing fully in the risky asset), (ii) Bank 2 invests all in the risky asset and Bank 1 invests enough in the risk-free asset to always be solvent even if Bank 2 defaults, and (iii) both banks invest fully in the risky assets.\footnote{One can imagine other cases where Bank 2 invests enough to pay back some of its debts but not to avoid bankruptcy, and then Bank 1 can lower the amount it needs to invest in the risk-free asset to remain solvent.  However, it is easy to check that this does not improve over whichever is the best of the three scenarios above.}

First, let us note that in any situation where (ii) is better than (iii), then (i) is better than (ii).  This implies that we can just consider (i) and (iii) as the social optima:   either prevent Bank 2 from ever defaulting or let both banks invest fully in the risky
asset.

To see that in any situation where (ii) is better than (iii), then (i) is better than (ii), consider what has to be true for (ii) to be better than (iii).   It has to be that the lost expected returns from requiring Bank 1 to invest enough in the risk-free asset to always be solvent is less than the expected bankruptcy cost.
This is true if and only if
$$(1-\theta) \chi  >\frac{D}{1+r}[\theta R - (1+r) ]  ,$$
where $(1-\theta) \chi$ is the expected bankruptcy cost, $\theta R - (1+r)$ is the lost expected returns,  and $\frac{D}{1+r}$ is the amount that Bank 1 must invest
in the safe asset to stay solvent.
If this holds, then the same calculation for (i) is that
\begin{equation}\label{bank2} 2(1-\theta) \chi  >\frac{2D}{1+r}[\theta R - (1+r) ] ,
\end{equation}
since Bank 2's solvency prevents two bankruptcy costs but has to cover twice the debt.
Whenever the first inequality holds, and (ii) improves welfare over (iii), then the second inequality must hold as well, and (i) leads to twice as much increase in welfare.

Thus, we only need to consider when it is that having both banks invest fully in the risky asset (option (iii)) is better than having Bank 2 invest $1-q_2 =\frac{2D}{1+r}  $ in the safe asset so as to prevent bankruptcies (option (i)). This is fully characterized by (\ref{bank2}), and is quite intuitive: the inequality compares the overall expected bankruptcy costs to the risk premium, scaled by the amount of debt that needs to be covered.

Next, let us examine the banks' incentives when they are maximizing their own shareholders' values.  If it is socially optimal to let both banks invest fully in the risky asset, it is easy to check that they will both do so.  Thus, we focus on the case in which it is best to have Bank 2 invest enough in the risk-free asset to avoid bankruptcy.

For the efficient portfolio profile to be an equilibrium when (\ref{bank2}) holds, it has to be that when Bank 1 chooses $q_1=1$, it is optimal for Bank $2$ to choose $q_2=1- \frac{2D}{1+r}$.  (If Bank 2 does so, then it is clearly optimal for Bank 1 to invest fully in the risky asset since it gives a greater expected return and there is no danger of bankruptcy.)
The expected value for Bank 2 from choosing $q_2=1- \frac{2D}{1+r}$ is
\[\left(1- \frac{2D}{1+r}\right)\theta R + \frac{2D}{1+r}(1+r) - 2D =  \left(1- \frac{2D}{1+r}\right)\theta R.\]
If Bank 2 invests fully in the risky asset, it however gets
\[\mathbb{E}[V_2^+(q_1=1, q_2=1)]=\theta (R  - 2D).\]
This is a profitable deviation, so the only equilibrium has both banks invest fully in the risky asset. This illustrates our first result (Proposition \ref{risky}).

Importantly, the inefficiency is not \emph{solely} driven by the fact that banks overlook their own bankruptcy costs and unpaid liabilities when choosing investments. To illustrate this, suppose instead that Bank 2 were to account for its bankruptcy cost $\chi$ and liability $2D$ upon default in its objective function, such that it maximizes $\mathbb{E}[V_2]$.  Investing fully in the risky asset then yields
\[\theta (R  - 2D) - (1-\theta)(2D+\chi).\]
 Thus, Bank 2 has efficient incentives  if and only if
\begin{align*}
  \left(1- \frac{2D}{1+r}\right)\theta R \geq \theta R  - 2D - (1-\theta)\chi,
\end{align*}
which simplifies to
\begin{equation}\label{opt2}
(1-\theta) \chi \geq \frac{2D}{1+r}\left[\theta R - (1+r)\right].
\end{equation}
Note the difference between the bank's incentive in (\ref{opt2}) and the socially efficient investment from (\ref{bank2}), which
differ by a factor of 2 on the left-hand side.   Hence the equilibrium can be inefficient even when banks internalize the ``direct" cost of their default, as captured by bankruptcy costs $\chi$ and missed payments $2D$. The inefficiency here comes from network externalities:
since a safer portfolio for Bank 2 has a positive externality on Bank 1, the socially efficient portfolio profile compares the opportunity cost of a safer portfolio with the \emph{overall} expected benefit $2(1-\theta)\chi$. However, when Bank 2 chooses its investment, it fails to internalize this externality and only compares the opportunity cost with \emph{its own} expected bankruptcy cost $(1-\theta)\chi$.

Next, let us examine the  incentives of banks to correlate their investments. To that end, we modify the network slightly: There are now two core banks (Bank 2 and Bank 3) and each owes an amount $D$ to the peripheral Bank 1. Each core bank also owes $0.5D$ to the other. The updated network is depicted in Figure \ref{fig:ex2}. Effectively, it is as if we had ``split'' our initial Bank 2 into two different banks that have claims on each other. The goal is to investigate whether these core banks have an incentive to correlate their investments.
As before, banks maximize their expected equity values $\mathbb{E}[V_i^+]$ when making investments.
\begin{figure}[!h]
\begin{center}
\begin{tikzpicture}[scale=1]
\definecolor{afblue}{rgb}{0.36, 0.54, 0.66};
\foreach \Point/\PointLabel in {(0,-1.5)/3, (0,1.5)/2, (4,0)/1}
\draw[fill=afblue!40] \Point circle (0.35) node {$\PointLabel$};
\draw[->, thick] (0.5,1.3) to node[above, sloped] {$D_{12}=D$} (3.5,0.2);
\draw[->, thick] (0.5,-1.3) to node[below, sloped] {$D_{13}=D$} (3.5,-0.2);
\draw[->, thick] (0.2,1.1) to  [out=-60,in=60] node[above, sloped] {$D_{32}=\frac{D}{2}$} (0.2,-1.1);
\draw[<-, thick] (-0.2,1.1) to  [out=-120,in=120] node[below, sloped] {$D_{23}=\frac{D}{2}$} (-0.2,-1.1);
\draw[->, thick] (4.5,0) to node[above] {$D_{01}=D$} (7.5,0);
\draw[fill=black!10] (8,0) circle (0.35) node {$0$};
  \end{tikzpicture}
  \end{center}
  \captionsetup{singlelinecheck=off}
  \caption[]{\small Banks 2 and 3 each have a liability of $0.5D$ to the other, and  a debt liability of $D$ towards Bank 1. Bank 1 has a debt liability of $D$ to outside investors (node 0). }\label{fig:ex2}
\end{figure}

There now exist three independently distributed risky assets whose realized gross returns
equal $R_1$, $R_2$, and $R_3$, respectively, with same probability $\theta$ and 0 with the residual probability, where $R_k\geq 1.5 D$ for each $k$. To make this motivating example as transparent as possible, there is no safe asset, and each bank can only choose one risky asset to invest in.\footnote{Similar conditions to the ones derived in the first part of the example guarantee that, even if a safe asset were available, banks would never choose to invest in it in equilibrium.} That is, a bank cannot divide its investment across different risky assets (we handle more complex cases in our analysis below). Suppose Bank 1 is fully invested in risky asset $k=1$ and let $\chi\leq 0.5 D$.

\begin{claim}\label{corrExample}
If the possible realizations of the risky assets $R_k$s are sufficiently close to each other, then it is socially efficient  for banks to diversify and all invest in \emph{different} assets (e.g., each Bank $i$ fully invests in risky asset $k=i$).
However, in all Nash equilibria, Banks 2 and 3 invest in the \emph{same} asset. Furthermore, for some parameter values, there are several such equilibria, including one in which both Banks 2 and 3 invest in the asset with the \emph{lowest} value $R_k$.
\end{claim}

The claim first states that diversifying investments across banks is socially efficient, whenever these investments have comparable expected returns. Why is this true? Since $R_k\geq 1.5 D$, a bank always remains solvent when its investment pays off, irrespective of which asset it invested in. More interestingly, Bank 1 also remains solvent if it receives (any) one of its debt payments from Banks 2 and 3. Hence diversifying investments across banks reduces the likelihood that Bank 1 defaults from $(1-\theta)$ to $(1-\theta)^3$. If expected returns are similar across risky assets, then there is no significant opportunity cost to choosing diversified portfolios, and hence doing so is socially efficient.

However, Banks 2 and 3 have a strict incentive to invest in the same risky asset. The intuition behind this part of the claim is more subtle, and comes from the fact that debt claims between banks generate complementarities between their values. To illustrate this, consider what happens when Banks 2 and 3 invest in different assets. For example, suppose that each Bank $i$ invests in asset $k=i$, and let $R_2 = R_3\equiv R$.\footnote{If returns are different across assets, then incentives to deviate are even stronger as whoever ends up investing in the asset with the lowest expected return has an extra incentive to choose the other's asset.}  If both assets pay off, then Banks 2 and 3's values are symmetric equal to $V_2(R,R) = V_3(R,R)=R-D$. If only Bank $i$'s asset pays off, however, then only Bank $i$ remains solvent. Indeed, the other gets its $0.5$ repayment from $i$, but this is not enough to cover its liabilities. The bank that remains solvent has then a value of $V^+_i(p_i = R, p_j = 0) = R+0.5D - \chi-1.5 D$ while the other defaults and gets $V^+_j(p_j=0,p_i= R) = 0$.
The expected value of each bank when they choose different assets is
 \begin{align*}
 \theta^2\underbrace{(R-D)}_{ \substack{\text{$V^+_i$ when both}\\\text{assets pay off}}}+\quad\theta(1-\theta)\underbrace{(R-D-\chi)}_{ \substack{\text{$V^+_i$ when only  $i$'s}\\\text{ asset pays off}}}< \underbrace{\theta(R-D)}_{\mathbb{E}[V^+_i]\text{ when $i$ invests in same asset as $j$}}.
 \end{align*}
 This illustrates our Propositions \ref{correlation1} and \ref{correlation2}. Banks have an incentive to correlate their investments not only to
minimize the probability of having to repay debts---which is reminiscent of the asset
substitution problem in corporate finance---but also to increase their
expected returns
from interbank assets conditional upon being solvent. Indeed, when $i$ correlates its investments with $j$, the expected debt repayment from $j$ to $i$ conditional on $i$ being solvent is the full repayment $0.5D$, whereas it is only $\theta 0.5 D$ when investments are not correlated. 

Finally, note that this incentive to correlate is strict. Hence, even if their counterparty is investing in the asset with the lower return, they can still prefer to do that too rather than invest in a different asset.
\end{example}

\section{Network Externalities and Distorted Incentives}

In this section, we study how financial contracts between banks distort investment incentives, because of associated externalities in insolvencies and bankruptcy costs.
The distortions in incentives are to be expected given the externalities present, but are important to characterize because of their
extremity.
This also provides the base for an analysis of regulation that
compares the performance of different regulatory instruments
in terms of endogenous investments and systemic risk consequences. We take as given the network of
financial obligations between banks throughout the analysis, and discuss this assumption in Section \ref{sec:endonetwork}.

\subsection{Overly Risky Investment:  The Intensive Margin}\label{overly}

We begin by examining the intensive margin of risk-taking. Each bank has $|K_i|=2$ assets in which it can invest. The first is a \emph{s}afe asset with deterministic return $p_s=1+r$; the second a \emph{r}isky asset that pays a random return $p_{ir}$ with  $ \mathbb{E}[  p_{ir} ]> 1+r$. One can think of this as a standard two-fund separation setting, with the main decision of the investor being how much risk to take. We make no assumptions on the joint distribution of risky asset returns $(p_{ir})_i$, which can be arbitrarily correlated across banks.

 Suppose each bank's outstanding debt $D_i^{L}$ is low enough that it could be paid back entirely were the bank to only invest in the safe asset: $D_i^{L} \leq (1+r)$.
Insolvency happens if the value of a bank's portfolio falls below the bank's liability $D_i^{L}$,
in which case bankruptcy costs  are incurred. Under limited liability, the shareholders get a payoff of zero in case of insolvency, and the bankruptcy costs are born by
whomever is holding debt, or a government  that steps in.

Network interdependencies could hypothetically generate market discipline, as a risky investment of Bank $i$ can trigger losses for some of its counterparties and feedback to itself. We show, however, that this is not the case and that all banks have a strict incentive to take on as much risk as possible.

\begin{proposition}\label{risky}
Suppose the distribution of each bank's risky asset return $p_{ir}$ is atomless. Then investing fully in the risky asset $q_{ir}=1$ is a strictly dominant strategy for all banks, irrespective of the network structure.

Furthermore, if $\Pr(p_{ir}<D_i^L-D_i^A)>0$, then there exists $\overline{a}$, $\overline{\chi}$ such that, if bankruptcy costs are large enough $ a\geq \overline{a}$, $\chi\geq \overline{\chi}$, then Bank $i$'s investment decision is inefficiently risky.
\end{proposition}

The proof of Proposition \ref{risky} involves showing that a Bank $i$ always benefits from marginally increasing its investment in the risky asset (up to a limit of full investment as we take available capital as given). Doing so leads to strictly higher expected returns on investment. It might trigger Bank $i$'s default in some states of the world, but this can only be the case if Bank $i$ is on the verge of insolvency in those states, such that the losses to owners are negligible. Bank $i$'s default might trigger other defaults that further decrease the value of its assets in discontinuous ways, but these losses do not affect the owners of Bank $i$'s equity value either, since $i$ is already defaulting and its owners are getting a payoff of zero. There is then a strict incentive to take on as much risk as possible.

Fully investing in the risky portfolio is often socially inefficient, since a bank's decision also
affects the rest of the financial network. First, a bank does not account for
its default and incurred bankruptcy costs when deciding its investment. Indeed, under limited liability, the bank's
shareholders only consider returns earned when solvent and completely disregard what happens under insolvency. Second, a bank's investment decision impacts others through counterparty risk. In particular, if $i$ defaults it will not honor its debt liabilities and its creditors may be driven to insolvency, causing bankruptcy costs to add up. Because of these, a planner would often prefer less risky investments.

The intuition behind this is straightforward.  The bank has incentives
to maximize
\[
\mathbb{E}_\mathbf{p}[ V_i(\mathbf{q_i},\mathbf{q_{-i}};\mathbf{p},\mathbf{D}) \, | \, V_i(\mathbf{q_i},\mathbf{q_{-i}};\mathbf{p},\mathbf{D}) > 0 ] \Pr[ V_i(\mathbf{q_i},\mathbf{q_{-i}};\mathbf{p},\mathbf{D}) > 0 ].
\]
Society's full payoff as a function of $i$'s portfolio choice
is however
\[
\mathbb{E}_\mathbf{p}\left[\sum_j \mathbf{q_j}\mathbf{p}
- b_j(\mathbf{V}(\mathbf{q_i},\mathbf{q_{-i}};\mathbf{p},\mathbf{D}),\mathbf{q_i},\mathbf{q_{-i}}, \mathbf{p},\mathbf{D})\right].
\]
Hence the bank enjoys the benefits of riskier investments coming from the risk premium,
but overlooks its externality when defaulting, and any additional defaults it can trigger.
This leads banks  to tend to make inefficiently  risky investments.

We emphasize that even though the above results are amplified by limited liability, the inefficiency is not dependent upon that limit.
Even if a bank maximizes $\mathbb{E}[V_i]$ and thus accounts for its expected failure costs and missed payments upon default, its default can cause defaults of other banks that are not part of its objective function.
Thus, even when fully accounting for its own potentially negative value, a bank can have incentives to over-invest in a risky
asset from a systemic-risk perspective (as illustrated in Section \ref{example}).

\subsection{Correlated Investments: Popcorn \emph{and} Dominoes}\label{corrdom}

The metaphor of ``popcorn or dominoes'' was made by Eddie Lazear, the chairman of the council of economic advisors under George W. Bush during the financial crisis.
The question was whether there really was any issue of potential contagion and ``dominoes'', or whether much of the crisis was instead simply due to all banks ``boiling in the same hot oil''---i.e., all having extensive exposure to an under-performing mortgage market.
The answer is that both were true.   Banks had highly correlated portfolios and all had dangerously
low values in their investments at the same time, and hence most were either barely solvent, or
even insolvent.  Nonetheless, they also had large exposures to each others' debts, as well as to
derivatives from AIG who could not even manage margin payments, and to securities issued by
Fannie Mae and Freddie Mac, which were both insolvent.\footnote{For discussion of this see Jackson \citeyearpar{jackson2019}, as well as the extensive analysis
and data in the {\sl Financial Crisis Inquiry Report}, commissioned by an act of the US congress.}
This highlights the fact that correlation of investments across banks matters for financial contagion:
many organizations holding
directly or indirectly similar subprime mortgages made the whole system substantially more
fragile.

We now investigate banks' incentives to correlate their investments.

\subsubsection{A General Result on Correlation and `Risk-Matching'}

There are $K=N$ assets available to all banks. Each asset yields a high gross return $p_{k}=\overline{R}$ with probability $\theta$ and a low gross return $p_{k}=\underline{R}<\overline{R}$ otherwise. Returns are independent across assets, but by choosing which asset to invest in, banks can effectively choose how to correlate their portfolios. For instance, if they all choose to invest in the same asset then their portfolios are maximally positively correlated, whereas if they all choose different assets then their portfolios are independent. Because we are interested in studying how banks correlate their portfolios rather than diversification incentives, we impose that each bank invests in only one asset; i.e., for all $i$, $q_{ik} = 1$ for some $k\in K$. Let $\overline{R}\geq D_i^L-D_i^A$ for all $i$, so that all banks remain solvent when they all receive a high return. If this condition does not hold, then some bank(s) always default irrespective of investments and realized returns, and we can redefine the network of interest to be the remaining banks with non-trivial solvency status.

To simplify notation in this section, we no longer make explicit the dependence of bank values $\mathbf{V}$ on the network of debt  $\mathbf{D}$, and we write bank values as a function of realized portfolio values $\mathbf{qp}$ only. Let $\mathbf{q}_{-i}\mathbf{p}=\mathbf{\overline{R}}$ denote that all banks other than $i$ have received high returns $\overline{R}$s, and
$\mathbf{q}_{-i}\mathbf{p}=\mathbf{\underline{R}}$ denote that they all have gotten low returns $\underline{R}$s.

\begin{proposition}
\label{correlation1}
Suppose that, for some Bank $i$, the value of a high portfolio realization is strictly higher when some Bank $j$ also has a high portfolio realization than when it does not---i.e.,
\begin{align*}
    \mathbb{E}_{\mathbf{q}_{-ij}\mathbf{p}}[V^+_i(\mathbf{q}_{i}&\mathbf{p}=\overline{R}, \mathbf{q}_{j}\mathbf{p}=\overline{R}, \mathbf{q}_{-ij}\mathbf{p})] -\mathbb{E}_{\mathbf{q}_{-ij}\mathbf{p}}[V^+_i(\mathbf{q}_{i}\mathbf{p}=\underline{R}, \mathbf{q}_{j}\mathbf{p}=\overline{R}, \mathbf{q}_{-ij}\mathbf{p})] \\&> \mathbb{E}_{\mathbf{q}_{-ij}\mathbf{p}}[V^+_i(\mathbf{q}_{i}\mathbf{p}=\overline{R},\mathbf{q}_{j}\mathbf{p}=\underline{R}, \mathbf{q}_{-ij}\mathbf{p})] -\mathbb{E}_{\mathbf{q}_{-ij}\mathbf{p}}[V^+_i(\mathbf{q}_{i}\mathbf{p}=\underline{R},\mathbf{q}_{j}\mathbf{p}=\underline{R}, \mathbf{q}_{-ij}\mathbf{p})]
\end{align*}
 for some $i$, $j$. Then there is no equilibrium of the investment game in which portfolios are jointly independent across banks.

 A sufficient condition for the above inequality to hold is that there exist two neighboring banks $i$, $j$ with $D_{ij}>0$ such that $\underline{R}\leq D_i^L-D_i^A<\overline{R}$ and $\underline{R}<D_j^L-D_j^A$.
\end{proposition}

Proposition \ref{correlation1} gives sufficient conditions for independent portfolios not to be part of any equilibrium. These conditions are easily satisfied whenever the low return $\underline{R}$ is low enough. Indeed, they require the existence of two neighboring banks in the network, such that the debtor Bank $j$ necessarily defaults when it gets a low return ($\underline{R}<D_j^L-D_j^A$) and its creditor $i$ either defaults as well or is on the verge of defaulting ($\underline{R}\leq D_i^L-D_i^A$). Whenever this is the case, Bank $i$ has a strict incentive to correlate its portfolio to Bank $j$'s.

This incentive to correlate goes beyond neighboring banks: Proposition \ref{correlation2} shows that fairly weak conditions are sufficient to ensure that
correlation of \emph{all} banks' portfolios is an equilibrium.

\begin{proposition}
\label{correlation2}
Suppose that, for all banks, the value of a high portfolio realization is weakly higher when all other banks also have a high portfolio realization than when they all have low realizations---i.e.,
\begin{align*}
    V^+_i(\mathbf{q}_{i}\mathbf{p}=\overline{R}, \mathbf{q}_{-i}\mathbf{p}=\mathbf{\overline{R}}) -V^+_i(\mathbf{q}_{i}\mathbf{p}=\underline{R}, &\mathbf{q}_{-i}\mathbf{p}=\mathbf{\overline{R}}) \\&\geq V^+_i(\mathbf{q}_{i}\mathbf{p}=\overline{R}, \mathbf{q}_{-i}\mathbf{p}=\mathbf{\underline{R}}) -V^+_i(\mathbf{q}_{i}\mathbf{p}=\underline{R},\mathbf{q}_{-i}\mathbf{p}=\mathbf{\underline{R}})
\end{align*}
 for all $i$. Then there exists an equilibrium in which banks' portfolios are perfectly correlated.

Sufficient conditions for the above inequality to hold is that, for each Bank $i$, either (i) $i$ is not part of a cycle, or (ii) $D_i^L\leq \underline{R}$, or (iii) $ D_i^L-D_i^A \geq \underline{R}$.
\end{proposition}

To get intuition for Proposition \ref{correlation2}, take the point of view of some Bank $i$ and consider what happens when all banks $j\neq i$ choose the same asset. If Bank $i$ does not choose that same asset as well, there will be states in which all banks but
$i$ get a low return. If high returns are (weak) complements,
then Bank $i$  prefers to receive its high return when others receive a high return as well. For the condition in Proposition \ref{correlation2} not to hold, Bank $i$ must gain strictly more by having a high return when others have low returns. This can only be the case if a high return for Bank $i$ can prevent others' defaults and feed back to $i$ in a way that strictly increases $i$'s net payoff. So $i$ would have to belong to a dependency cycle (condition (i) not holding), and $i$ would need to be sufficiently integrated in the network (conditions (ii) and (iii) not holding). In addition, it would need to receive more from its counterparties than what it pays to them itself when $\mathbf{p}_{-i}=\underline{R}$. Only a combination of all the above could lead to a strict incentive to diversify investments from one's counterparties.

We refer to this incentive to correlate as `{\sl risk matching}' to distinguish it from
risk shifting.  Risk shifting refers to the general phenomenon that someone who faces some form of limited liability and/or partial returns has incentives to distort investments to ones that may be inefficient (e.g.,  Galai and Masulis \citeyearpar{galai1976option} and Jensen and Meckling \citeyearpar{jensenm1976}).
Basically, an investor has an incentive to arrange
a portfolio so that they earn the most returns, even if that results in
increased risks for other parties not involved in the portfolio choice, and even if this involves lower overall value investments.
Our setting does have this feature, as we have already demonstrated.  However, what is special in Proposition \ref{correlation2} is that the distortion here comes from incentives {\sl across} banks to correlate their investments due to complementarities in their balance sheets and the values received.
Banks have a special incentive to match their investments.  This is a feature that is not present in what is usually thought of as risk shifting, and so we give it a new name.
We note that the result in Elliott, Georg, and Hazell  \citeyearpar{elliottgh2018} is an example of risk
shifting, as banks want to correlate their assets to shift losses from
states in which the bank is solvent and the loss is incurred by
\emph{shareholders}, to states in which the bank defaults and the loss
is incurred by \emph{debt holders}.
This force is present in our setting as well but the incentive to correlate goes beyond it, since it applies even when banks account for their own failure costs and missed payments, as illustrated in Example \ref{ex:1} from Section \ref{example}.
The incentive to correlate here comes from the fact that high portfolio
realizations are
complements: a bank generally gains more by remaining solvent and getting
$p_i=\overline{R}$ when others also have high
portfolio realizations since its own value depends positively on others'
through financial interdependencies. 
This intuition holds regardless of how bankruptcies are resolved
or how large those costs are, and in particular without assuming that a bank bears the costs of
its counterparty's bankruptcy.\footnote{We also note that the intuition relies on network externalities generated by interbank contracts, which differ from price-based externalities studied in Acharya \citeyearpar{acharya2009theory}.}

We take the interbank network as given in our analysis, but several existing papers show that, if banks can choose their counterparties, they have an incentive to correlate their counterparty exposures as well (Elliott et al Elliott, Georg, and Hazell  \citeyearpar{elliottgh2018}; Erol and Vohra \citeyearpar{erolv2022}). The equilibrium network then has a clique structure, where banks in the same clique are highly connected to each other. If anything, this should only strengthen the incentive to correlate outside investments with other clique members.

\paragraph{Risk Aversion}

Although we have worked with risk-neutral banks to keep the analysis uncluttered and to emphasize the impact of financial interdependencies, it should be apparent that some incentive to correlate portfolios persists when investors are risk averse. An equilibrium in which all banks choose the same portfolio exists whenever the inequality in Proposition \ref{correlation2} holds, but with each $V_i^+$ transformed by some concave function $u_i$. This transformation makes the inequality harder to satisfy, but it still holds, for instance, if a bank cannot remain solvent when all banks but itself get a low return. In that case, the RHS is zero while the LHS is positive, irrespective of risk aversion.

\subsubsection{Uniqueness of the Full Correlation Equilibrium}

Proposition \ref{correlation3} provides insight into two of the most natural types of correlation, but does not address all possible correlation structures.
Generally, banks have incentives to line up their portfolio realizations, so as
to maximize their expected return from interbank assets. Nonetheless,
one does need stronger conditions to ensure that all banks want to perfectly correlate their portfolios in {\sl all} equilibria - such a strong form of uniqueness requires ruling out equilibria in which the network is partitioned into subsets of banks that correlate their portfolios within each subset, but choose (partially) uncorrelated portfolios across subsets.
As we show next, under stronger conditions, perfect correlation of portfolios is the unique equilibrium.

\begin{proposition}
\label{correlation3}
If the value from a high portfolio return  is
increasing in the number of high portfolio returns among other banks---i.e.,
\[V_i^+(\mathbf{q}_{i}\mathbf{p}=\overline{R}, \mathbf{q}_{-i}\mathbf{p}) -V^+_i(\mathbf{q}_{i}\mathbf{p}=\underline{R}, \mathbf{q}_{-i}\mathbf{p}) > V^+_i(\mathbf{q}_{i}\mathbf{p}=\overline{R}, \mathbf{q}_{-i}\mathbf{p}') -V^+_i(\mathbf{q}_{i}\mathbf{p}=\underline{R}, \mathbf{q}_{-i}\mathbf{p}')\]
 for each $i$, $\mathbf{q}_{-i}\mathbf{p}$, $\mathbf{q}_{-i}\mathbf{p}'$
 such that $|\{j\neq i: \mathbf{q}_{j}\mathbf{p}=\overline{R}\}|>|\{j\neq i:  \mathbf{q}_{j}\mathbf{p}'=\overline{R}\}|$, then \emph{any} equilibrium must feature perfect correlation of portfolios across banks.
\end{proposition}

A sufficient condition for perfect correlation to be the unique equilibrium outcome
is that the marginal gain in equity value from a high realization of one's own portfolio is strictly
increasing in the number of other banks that also have a high portfolio realization. This requires that no bank would
prefer to correlate with some particular set of counterparties as opposed to some other larger set of banks.

For example, consider a network in which each bank owes $D>0$ to every other bank and $D_{0}\geq \underline{R}$ to outside depositors. Furthermore, let $\overline{R}>D_i^L=(n-1)D+D_0$ for each $i$, such that contagion is precluded.   In that case, the equity value of a bank with low return is always zero, while the value of a high return is strictly increasing in the number of other banks with high returns, thus satisfying the above condition.\footnote{The condition in Proposition \ref{correlation3} also holds in settings in which there is contagion if sufficiently many banks fail together, but such contagion is stopped by a regulator via bailouts. If such bailouts are anticipated and if solvent banks expect greater partial payments from insolvent banks when fewer banks are insolvent, then perfect correlation of investments is also the unique equilibrium.  Note that, with such bailouts, the reasons for wanting to correlate are different than those noted by \cite{acharya2007}, where banks correlate portfolios to induce bailouts.  Here, in the condition just mentioned, banks would expect bailouts and wish to correlate portfolios in order to be solvent when more banks make full payments on debts as they then enjoy higher returns.}

Without such a condition, there can exist other equilibria in which there is partial correlation of portfolios. 
Note that even when there exist other equilibria with partial correlation,
all banks get their highest possible payoff in a full correlation equilibrium.
Hence if they can jointly deviate and are not getting their highest possible expected payoff, then they would prefer to move to a perfect correlation equilibrium.

\subsubsection{The Inefficiency of Full Correlation}

In terms of efficiency,
maximizing the total value of all private investors in the economy is equivalent to minimizing the expected number of defaults.
Indeed, the correlation structure of investments across banks does not change the expected aggregate portfolio value $\sum_i \mathbf{q}_{i}\mathbf{p} =  N(\theta\overline{R}+ (1-\theta)\underline{R})$,  but it does impact the set of defaulting banks and hence the amount of bankruptcy costs incurred. Correlated investments across banks are then socially efficient if and only if they induce a lower expected number of defaults than some other configuration.
This holds if any bank that gets a low return always becomes insolvent irrespective of what else happens to other portfolios: independent investments do not attenuate systemic risk since high returns from some banks can never prevent another from defaulting. In that extreme case, correlated investments are socially efficient, and banks' incentives are aligned with that of the social planner.
However, as soon as correlation worsens contagion risk, the equilibrium is generally not socially optimal.
This will be true in many cases of interest, such as when a bank that gets a low return can
maintain its solvency by receiving its debt payments from its counterparties. Thus, in most cases of interest the decentralized equilibrium is socially inefficient.

\section{A Network Approach to Measuring and Regulating Risk-Taking}

Network externalities lead banks to take on too much risk and correlate the realizations of those risks, as they generally do not internalize how their investment decisions affect other organizations in the system. We now explore more deeply what these network externalities imply for optimal regulation.\footnote{In what follows, we suppose that the regulator knows what the network looks like. For an analysis of regulation in the face of uncertainty about banks' positions and network structure, see Ramirez \citeyearpar{ramirez2019}.}
We focus on the optimal regulation of the intensive margin, as it involves a real tradeoff between higher expected returns on investments and greater systemic risk. 
Banks' incentives to overly correlate their investments can be resolved directly by restricting the overlap between their portfolios.\footnote{Within our framework, the correlation of banks' investments only enters social welfare through its effect on expected insolvencies. Optimal regulation should then require whichever correlation level minimizes expected insolvencies. However, this overlooks the (unmodeled) economic benefits that arise from having multiple banks invest in the same assets (e.g., large investments that cannot be funded by just one bank). A fuller treatment of the correlation structure of banks' portfolios is left for further research, as it requires modeling these benefits.}

In general, characterizing optimal regulation is a complex problem as the value of regulating a particular bank's investments depends on how other banks are regulated. It is then necessary to consider all different subsets of banks to regulate, and compare the overall societal value in each case. Given the complexity of this problem, we first focus on characterizing optimal regulation in two classes of core-periphery networks. We then provide a broader framework and insights into optimal regulation for general networks.

\subsection{Regulation with Correlated Investments and Core-Periphery Networks}\label{sec:coreperiphery}

Core-periphery networks are empirically relevant network structures,\footnote{For example, the derivatives network between the large U.S. bank holding companies features a core-periphery structure (D'Erasmo, Erol and Ordonez \citeyearpar{d2022regulating}).} and they highlight how regulation can be asymmetric.
Importantly, we allow for correlation between returns to investment opportunities across banks, as this turns out to be a key consideration in
the optimal regulation.

\subsubsection{Symmetric Core-Periphery Networks}\label{core}

Consider a core-periphery network consisting of $n_c$ core banks and $n_p$ peripheral banks, where $n_p$ is a
multiple of $n_c$.   Core banks form a clique: they all have a debt claim and a debt liability of $D$ on each other.
Each core bank is linked to $n_p/n_c$ peripheral banks, to which it owes a total of $D_0$, and each peripheral bank is linked to just one core bank.
Peripheral banks have no additional investments: they just act as intermediaries between core banks and depositors, such that $D_{0i}=D_i^L=D_i^A$ for each peripheral bank $i$. Intuitively, this is a setting where regional banks collect deposits, send those along to core banks, which make the investments, and then return an amount back through the peripheral banks to depositors.  The results directly extend to having more (or no) peripheral banks.
See Figure \ref{fig:core_peri} for an illustration with $n_c=n_p=4$.
\begin{figure}[!h]
\centering
\includegraphics[width=0.6\textwidth]{core_peri_example.tikz}
\caption{\small
Arrows point in the direction that debts are owed. There are four core banks (dark/red) linked together via debt claims of $D$ on one another. There are four peripheral banks (light/blue), each having a debt claim of $D_0$ on a single core bank and a liability of $D_0$ to depositors. Peripheral banks act as intermediaries between depositors and core banks. }\label{fig:core_peri}
\end{figure}

Let $a=0$ and $\chi>0$,  which is large enough so that no payments are made by a defaulting core bank---alternatively, payments are delayed enough so that creditors become insolvent as well.
This cost includes the costs of the bankruptcy of its associated peripheral banks.
Importantly, this cost is incurred by someone even if the bank has limited liability, and so is a deadweight loss. In what follows, we focus on the best-case equilibrium for bank values, but the analysis easily extends to other equilibria.

Core banks can either invest in a risk-free asset or in their own proprietary asset. That is, there are $K=n_C+1$ assets and $K_i = \{i,n_C+1\}$, where asset $k=i$ is $i$'s proprietary asset and asset $k=n_C+1$ is the risk-free asset. Each bank's risky asset  yields $p_i=R$ with probability $\theta<1$ and zero otherwise, with the maintained assumption that $\theta R>1+r$ so that the risky asset is not dominated by the risk-free asset.

Given the multi-dimensional return vector, we examine a natural parameterization in which we can clearly order joint-distress probabilities.
In particular, risky investments are correlated across banks as follows:
\[
\begin{cases}
\text{all banks get $p_i=R$ with probability }1-\frac{n_c(1-\theta)}{m},\\\text{$m$ banks get $p_i=0$ and $n_c-m$ get $p_i=R$ with probability }\frac{n_c(1-\theta)}{m},
\end{cases}
\]
 where $\theta\geq 1- \frac{m}{n_c}$ ensures that these are well-defined.  Moreover, the $m$ banks that get the 0 realization at the same time are chosen uniformly at random,
 so that every combination of core banks getting 0 realizations is equally likely.

 Thus, either all banks get a high investment realization, or $m$ randomly picked core banks get zero. Since the overall probability of getting a low realization for each bank is fixed, the larger the set of banks that get zero together ($m$), the smaller the probability of that event $n_c(1-\theta)/m$.
 So there is a trade-off between having fewer banks getting a zero return at once, but more often, or having many banks getting zero return together but with smaller probability. In this specification, $m$ captures the correlation in returns to risky investments across banks: in the extreme case where $m=n_c$, then all core banks effectively have the same investment, and either all banks get a high return or all banks get a zero gross return. 

Here we take the level of correlation $m$ as exogenous. However, the insights from Section \ref{corrdom} hold in this setting: core banks have a strong incentive to correlate their risky investments. Thus, one can think of $m$ as the maximal level of correlation banks can achieve, which is then the one they choose.

We presume that $R< D_0+(n_c-1)D = D_i^L$, so that a bank needs at least some of its debt to be paid back to be solvent, even if its risky asset pays off.
This ensures that there is some probability of contagion, as otherwise the joint regulation problem is uninteresting.

Let $k^R$ be the maximum number of defaults that a core bank can sustain when it gets a high investment realization and still remain solvent:
 \[R+(n_c-1-k)D\geq (n_c-1)D+D_0\iff k \leq k^R\equiv \lfloor\frac{R-D_0}{D}\rfloor.\]

Absent regulatory intervention, all banks invest fully in the risky asset.\footnote{See Appendix A for a formal proof.}  They then remain solvent when they all get $p_i=R$. However, if $m>k^R$, they all default in the state of the world in which $m$ of them get zero realizations.

\begin{remark}
Without any regulation, if $n_c-2\geq k^R\geq (1-\theta)n_c$,\footnote{This ensures that it is possible to have $m \leq k^R$ (the second inequality from the assumption that $\theta\geq 1- \frac{m}{n_c}$), and have a region in which there is a decreasing  expected number of defaults (the first inequality).}  then the expected number of defaults is non-monotonic in the amount of correlation $m$, and reaches a maximum at  $m=k^R+1$.
\end{remark}

When the correlation of investments is sufficiently small (i.e.,  $m\leq k^R$),
then the number of banks who get zero realizations and default is never large enough to cause those who get a high realization
to default. Then the expected number of defaults simply equals the expected number of banks who get a zero realization; that is, $(1-\theta)n_c$.
At the other extreme where $m=n_c$ there is perfect correlation in banks' investment realizations and so again there is no contagion.
The contagion occurs in the intermediate region in which $ k^R< m< n_c$ and then decreases in $m$, as pictured in Figure \ref{fig:defaults}.
When banks' assets are  correlated enough, sufficiently many banks get zero return at the same time to trigger a cascade leading all banks to default.  This happens with a probability that is decreasing in $m$, so systemic risk is highest when $m$ is just above the threshold $k^R$: investments are correlated enough to trigger a default cascade if they fail, but not so correlated that the likelihood of such an event is small.  This impacts the optimal regulation.

\begin{figure}[!h]
\begin{center}
\begin{tikzpicture}[scale=1]
\draw[->, thick] (-0.1,0)--(6,0) node[right, below=0.5em]{$m$};
\draw[->, thick] (0,-0.1)--(0,3) node[above, left]{Expected  \#};
\node at (-1.3,2.5) {of defaults};
\foreach \Point/\PointLabel in {(0,1), (0.5,1), (1.5,1), (1,1), (2,1), (2.5,2.5), (3,2.2), (3.5,1.9), (4,1.6), (4.5, 1.3), (5,1) }
\draw[cyan!50!black, fill=cyan!50!black] \Point circle (0.05);
\draw[thick] (2,-0.1)--(2,0.1) node[below=0.5em]{$k^R$}; \draw[thick] (5,-0.1)--(5,0.1) node[below=0.5em]{$n_c$};
  \end{tikzpicture}
  \end{center}
  \caption{\small
Expected number of defaults as a function of the correlation of investments $m$. }\label{fig:defaults}
\end{figure}

We break the analysis into two cases.  The first case is when the correlation between investments is low enough such that there is no possibility of contagion:  that is, $m\leq k^R$.
The second case is one in which it is possible that sufficiently many core banks get no return on their risky investments so that, if unregulated,
their defaults would lead other core banks to also become insolvent, even if they had high returns.

The first case is covered in Proposition \ref{symmetric}. We first consider only one type of regulatory intervention: macroprudential policy that restricts banks' investments in the risky asset.  We discuss the use of bailouts further below.

\begin{proposition}
\label{symmetric}
Suppose that $m\leq k^R$. Regulation is optimal if and only if\footnote{This corresponds to $\chi> \frac{D_0[R\theta-(1+r)]}{(1+r)(1-\theta)}$ and so holds for higher levels of bankruptcy costs.}
\[\theta< \frac{\chi+D_0}{\chi+\frac{R}{1+r}D_0}.\] When some regulation is optimal, it is symmetric and has all core banks invest $\frac{D_0}{1+r}$ in the risk-free asset and the remainder in the risky asset.
\end{proposition}

Proposition \ref{symmetric} is straightforward.  First, note that peripheral banks are only intermediaries, and so regulation concerns core banks.
 When $m\leq k^R$, a core bank that gets a high investment realization never defaults and so there is no contagion inside the core.  The only contagion comes from a core bank's default leading to failure of payments to its peripheral banks. Thus, optimal regulation is to regulate each core bank independently -- as if it were a stand-alone bank just with its peripheral banks and indirect depositors. Optimal regulation is then necessarily symmetric: it treats all core banks in the same way, and regulates them if and only if the risky asset is ``too'' risky; i.e., the probability $\theta$ that it pays off is too low.

The more complex, and potentially asymmetric, regulation arises when there is sufficient correlation among core banks' investments.
This is characterized in Proposition \ref{asymmetric}.

Let $k^r\equiv \lfloor \frac{1+r-D_0}{D}\rfloor $ be the maximum number of defaults that a core bank can sustain when it is investing fully in the risk-free asset.

\begin{proposition}
\label{asymmetric}
If $m>k^R$ and $k^r\geq 1$, then regulation is optimal if
\begin{equation}
\label{thetahigh}
\theta< \overline{\theta}\equiv \frac{\frac{n_c}{m}\chi+D_0}{\frac{n_c}{m}\chi+\frac{R}{1+r}D_0}.
\end{equation}
If regulation is optimal, then either (i) it is symmetric and has all core banks invest $\frac{D_0}{1+r}$ in the risk-free asset and the remainder in the risky asset,
or (ii) it is asymmetric.
In particular, if $ \left(n_c -k^r\right)(D_0+k^rD)<n_c$ and
\begin{equation}
\label{thetalow}
\theta > \underline{\theta} \equiv  \frac{ n_c-\left(n_c-k^r\right)(D_0+k^rD)+k^r\chi}{\left[n_c-\left(n_c-k^r\right)(D_0+k^rD)\right]\frac{R}{1+r} + k^r\chi},
\end{equation}
then optimal regulation is necessarily asymmetric.  That is,
some core banks face different limits on risky investments than other core banks
even though all core banks are ex ante identical in all ways.
\end{proposition}

Note that (\ref{thetahigh}) and (\ref{thetalow}) together correspond to\footnote{We also assume throughout this analysis that bankruptcy costs are large enough so that a defaulting bank pays back none of its debt. This requires $\chi\geq R+(n_c - 2 - k^R)D$ since a defaulting bank can have assets of value up to $R+(n_c - 2 - k^R)D$. Note that these conditions are not mutually exclusive---the ones in Proposition \ref{asymmetric} depend on $\theta$ while this one does not. }
$$\left(n_c-\left(n_c-k^r\right)(D_0+k^rD)\right) \frac{R\theta-(1+r)}{(1+r)(1-\theta)} > \chi>\left(\frac{m}{n_c}\right) \frac{D_0[R\theta-(1+r)]}{(1+r)(1-\theta)}. $$

When $m>k^R$, banks that get a high return $p_i=R$ become insolvent when $m$ of their counterparties get zero realizations and default on their payments. The possibility of such default cascades makes regulatory intervention more likely to be optimal.  That is,  $ \overline{\theta}$ is larger than the corresponding term in Proposition \ref{symmetric}.

Most importantly,
optimal regulation can often (i.e., for an open set of parameters) be asymmetric: when the probability that the risky asset pays off is neither too high nor too low ($\underline{\theta}<\theta<\overline{\theta}$), the regulator should restrict the investments of only a subset of core banks so as to reduce the likelihood of a default cascade, while still allowing some banks to benefit from the higher returns of the risky asset.
This is true even though the network and core banks are all fully identical ex ante.
The intuition is subtle, as one could imagine that allowing an intermediate investment level for all banks would be optimal, rather than asymmetric regulation.
The idea is that if instead one tried to impose symmetric regulation, then (given that $k^r\geq 1$) regardless of the level of that regulation, any core bank getting positive risky returns could withstand some
number of defaults (at least $k^r$) without having any contagion.  Then allowing $k^r$ banks to invest freely ends up increasing their expected returns, without increasing the risk of contagion.
Roughly, the level of total expected return has some convexity in the level of regulation, making it better to split the regulation to high and low levels rather than having all banks at an
intermediate level.  It is multidimensional, so that intuition is rough, but captures some aspects of the proof.

We emphasize that Proposition \ref{asymmetric} does not rely on indivisibilities in the network. Similar results would arise in a model with a continuum of core banks, where $m$ represents the share of core banks that get a zero return and $k^R$ the share of defaults a bank can sustain when it gets a high return. Asymmetric regulation that only imposes restrictions on some fraction of core banks would outperform symmetric regulation under the same conditions as in Proposition \ref{asymmetric}.

\begin{example}
We illustrate the intuition behind asymmetric regulation with the following example. There are $n_c=3$ core banks, and parameters are as follows: $m=2$, $k^R=k^r=1$. That is, the return distribution of risky assets is such that $m=2$ banks might get a zero return at the same time, but any bank with positive return can only sustain the default of one of its counterparties. If all banks invest in the risky asset, there is thus a risk of contagion: with probability $(3/2)(1-\theta)$, \emph{all} banks default, including whichever bank got a high return. The best symmetric regulation that prevents contagion requires each bank to invest $D_0/(1+r)$ in the safe asset. Such symmetric regulation improves over laissez-faire whenever
\[3\left[\theta R\left(1-\frac{D_0}{1+r}\right)+D_0\right]>3\left[\theta R - \frac{3}{2}(1-\theta)\chi\right]\iff \frac{D_0}{1+r}[\theta R - (1+r)]<\frac{3}{2}(1-\theta)\chi.\]
How can asymmetric regulation improve over symmetric regulation? Note that contagion is prevented if at most one bank gets a zero return and defaults at once. A natural policy candidate is then to ensure the solvency of two banks, by having them invest only in the safe asset, while allowing the third bank to invest fully in the risky asset. That third bank will sometimes default, but its default will not bring others to insolvency, and thus will not spread through the network. Total surplus associated with this asymmetric regulation is then
\[2(1+r) + \theta R - (1-\theta)\chi.\]
This improves over the above symmetric regulation whenever
\[(1-\theta)\chi< \left[3\frac{D_0}{1+r}-2\right][\theta R - (1+r)],\]
that is, whenever bankruptcy costs are small enough that they are worth risking absent contagion. \qed

\end{example}

If bailouts are possible at a cost $c$, then they are always optimal ex post whenever $c<\chi$. In that case, bailout costs simply replace bankruptcy costs. The more interesting case is when $c>\chi$, that is, when bailing out a bank is not worth it if it only prevents its own default.  Bailouts are then never optimal when correlation across banks' risky investments is low ($m\leq k^R$), since a bank defaulting never triggers another's default. Bailouts can however be optimal when the correlation across banks' investments is high ($m>k^R$). They are then necessarily asymmetric: the regulator only bails out sufficiently many banks so as to prevent a default cascade, and lets the others default.

The above analysis makes it clear that optimal regulation very much depends on how correlated banks' investments are. We however restricted banks' portfolio choices to either invest in a risky or a risk-free asset, and did not examine banks' options of correlations.
When they do have such options, one can still regulate them
simply via portfolio restrictions that then guarantee certain levels of solvency by limiting exposures even in the face of large correlations, as described above.
However,  in such cases, there can be improvements from controlling correlation without forcing less risky investments: banks can be allowed to bear greater risks, provided that those risks are less correlated with those of other banks.

Finally, optimal regulation also depends on the network structure. In particular, laissez-faire is less likely to be optimal when claims between banks are larger.

\begin{corollary}
As the size of interbank claims $D$ increases, the set of parameter values under which some regulatory intervention is optimal expands.
\end{corollary}

 Indeed, larger interbank claims imply that core banks can sustain fewer defaults before being dragged down to insolvency---i.e., $k^R$ is lower. Thus, default cascades are more easily triggered and some regulation more likely to be optimal to prevent them.

\subsubsection{Nested-Split Graphs}

We now consider a more general class of core-periphery networks that allows for asymmetries between core banks.  Such asymmetries are observed empirically, so we want to allow for them in our analysis.\footnote{For instance, \cite{d2025unintended} estimate that derivative exposures between the largest US banks have been shifting toward a tiered structure since the implementation of clearing regulations from the Dodd-Frank Act.} 
The setup is the same as above, except that core banks no longer form a complete clique but a nested-split graph. Core banks then differ in how connected they are to other core banks, creating a hierarchy within the core.

Formally, the set of core banks is partitioned into tiers $\ell=1,\dots, L$, such that banks in tier $L$ are connected to all other core banks, and banks in tier $\ell$ all have the same set of core counterparties, which is a strict superset of the core-counterparties of banks in tier $\ell'<\ell$. By construction, banks in tier $\ell$ then have more core counterparties than banks in tier $\ell'$ whenever $\ell>\ell'$.\footnote{The tiers can be characterized by their degrees, with banks in tier $L$ having degree $n_c-1$ and then core banks in lower tiers having lower degrees. If a core bank is connected to some other core bank with degree $d$, then it is connected to all core banks with degree at least $d$.}
These network structures have the property that the core can be partitioned into two sets: a clique of nodes that are all exposed to each other, and an independent set that is only exposed to the clique. That is, there exists a threshold $\overline{\ell}$ such that banks in tiers $\ell>\overline{\ell}$ form a clique, and banks in tiers $\ell\leq \overline{\ell}$ form an independent set.\footnote{Note that tiers below $\overline{\ell}$ might still have different connectivity to tiers above $\overline{\ell}$.  So, with multiple tiers, these can be quite rich objects. }

As before, if two core banks are counterparties, they both have a claim and liability of $D$ to each other. Higher-tier banks are then more ``central'' in the sense that they are connected to more core banks. They also have larger gross balance-sheets as they have greater total assets and liabilities. The previous section looked at the special case of a single tier, $L=1$.

Following a similar logic as above, each core bank can sustain up to $k^R\equiv \lfloor\frac{R-D_0}{D}\rfloor$ defaults when it gets a high realization, irrespective of the tier $\ell$ to which it belongs. As before, the expected number of defaults absent regulation is non-monotonic in the amount of correlation $m$. If few banks get zero return at the same time ($m\leq k^R$),  there is no contagion and only these banks default. Optimal regulation then remains symmetric.

\begin{proposition}
\label{prop:optregulation_NS}
If $m\leq k^R$, then regulation is optimal if and only if
\[\theta< \frac{\chi+D_0}{\chi+\frac{R}{1+r}D_0}.\]
When some regulation is optimal, it is symmetric and has all core banks invest $\frac{D_0}{1+r}$ in the risk-free asset and the remainder in the risky asset.
\end{proposition}

As correlation increases ($m>k^R$), so does the extent of contagion in the network, but in contrast to the case in which the core was a clique, contagion can occur in waves when there are multiple tiers. Absent regulation, the $m$ banks with zero realization necessarily default, which depresses the balance-sheets of those banks' counterparties.  If some core banks are exposed to sufficiently many of these $m$ initial defaults, they might default as well. Since banks in higher tiers are exposed to larger sets of banks, they are more at risk of contagion. In particular, banks from tier $L$ are, by construction, exposed to all other core banks, and so must default too. Whether contagion stops there or continues depends on the size of tier $L$ and on precisely which banks in the core got zero realization. For instance, if tier $L$ contains more than $k^R$ banks, then contagion spreads to the entire network.

Let $n_c^\ell$ denote the number of core banks in tier $\ell$.

\begin{proposition}
\label{prop:optregulation_NS2}
Suppose that $m>k^R$ and $k^r \leq \sum_{\ell\leq \overline{l}}n_c^\ell$ (so there are at least $k^r$ banks in the independent set).

If
\[\theta >\frac{\chi+D_0-\frac{(1+r-D_0)}{k^r}\sum_{\ell> \overline{l}}n_c^\ell}{\chi+\frac{R}{1+r}(D_0-\frac{(1+r-D_0)}{k^r}\sum_{\ell> \overline{l}}n_c^\ell)},\]
then a (nontrivial) regulation that treats core banks symmetrically is outperformed by one that imposes stricter limits on investments on higher-tier core banks while allowing some lower-tier banks to invest freely.
\end{proposition}

Intuitively, by regulating higher-tier core banks more than lower-tier ones, the regulator can insulate the network from contagion at a lower opportunity cost in terms of foregone risk premia. Consider for instance the following regulation. All core banks in tiers $\ell>\overline{l}$---that is, all core banks in the clique---are required to only invest in the risk-free asset. Core banks in the independent set are less restricted: $k^r$ of them face no restriction at all and invest freely, while the remaining have to invest $D_0/(1+r)$ in the risk-free asset. By construction, these latter banks  have enough reserves to cover their debts to the periphery, and so they remain solvent whenever they get their debt payments from their core counterparties. Importantly, they are only exposed to core banks from the clique. Banks from the clique are fully regulated and remain solvent as long as at most $k^r$ of their counterparties default. Since only $k^r$ core banks invest freely, there is never any contagion: even if these $k^r$ banks get low realizations and default, they are not very central and only exposed to the fully regulated clique, which stops contagion in its nest.

Proposition \ref{prop:optregulation_NS2} shows that, in the case of a hierarchical core, symmetric regulation is dominated by a particular form of asymmetric regulation in which the banks in higher tiers are regulated and then some number of the banks in lower tiers are allowed to invest freely.

Even though a hierarchical form of regulation outperforms symmetric regulation, that does not mean that the optimal (asymmetric) regulation is necessarily one in which the highest-tier banks are those
that are regulated.  The issue is that the probability of contagion differs across banks and in ways that depend on which banks are regulated.
The following example illustrates the issue.

\begin{example}\label{ex:complexity}
Consider a core with $L=2$ tiers. The high tier is comprised of $n_c^L=2$ banks, which form a clique and are connected to all core banks. The low tier is comprised of  $n_c^1=3$ banks that are just connected to the first two high-tier banks.
\begin{figure}[!h]
\centering
\includegraphics[width=0.4\textwidth]{NS_example.tikz}
\caption{\small
Arrows point in the direction that debts are owed. There are two core banks in the high tier (dark pink) which are exposed to all other core banks. The three core banks in the low tier (light pink) are only exposed to core banks from the high tier. Each core bank also has a liability of $1$ to $n_P$ peripheral banks (blue). To keep the graph uncluttered, we only drew the peripheral banks for one core bank.}\label{fig:NS}
\end{figure}
Let $m=3, k^R=2,  k^r=1$, so when some banks get the low realization it is a group of three that get it, and high realization banks can withstand 2 defaults while banks investing in the risk-free asset can withstand one default.

First, consider regulating the two clique banks and requiring that they invest only in the safe asset, so that they remain solvent as long as at most one of their counterparties defaults. Banks from the independent set are unregulated; they only have two core counterparties so they always remain solvent when $p_i=R$ and always default when $p_i=0$. Banks in the clique then remain solvent if and only if at most one bank from the independent set gets $p_i=0$, which happens with probability 3/10 conditional on $m$ banks getting $p_i=0$. They default otherwise. Thus, the social surplus is
\[\underbrace{2\left[(1+r)-\frac{n_c(1-\theta)}{m}\frac{7}{10}b\right]}_{\text{surplus from the two regulated banks}}\qquad\quad +\underbrace{3[\theta R - (1-\theta)b]}_{\text{surplus from the three unregulated banks}}.\]
Next, consider a different regulatory policy under which the clique is unregulated, but 2 out of the 3
 banks in the independent set are required to invest fully in the safe asset. These regulated banks default only when the two clique banks get $p_i=0$, which happens with probability 3/10 (conditional on $m$ banks getting $p_i=0$). The social surplus is
\[2\left[(1+r)-\frac{n_c(1-\theta)}{m}\frac{3}{10}\chi\right]+3[\theta R - (1-\theta)\chi],\]
which is higher, and so the second asymmetric regulatory policy improves over the first one.

Intuitively, banks in the independent set have fewer core counterparties, and so if regulated are less likely to be brought to insolvency by their counterparties. They are then more efficient as contagion buffer than high-tier banks, who are exposed to all other banks, and are most easily affected by contagion. \qed
\end{example}

Example \ref{ex:complexity} hints at the complexity of finding optimal regulation in general networks, which we now turn to.\footnote{This is reminiscent of the complexity of optimal bailout policies identified in \cite{jacksonp2020}.} We start by providing a measure of financial centrality that can guide regulatory decisions.

\subsection{A Measure of Financial Centrality}

We provide a network-based measure of financial impact of a given organization.
Conceptually, given our approach, there is a unique and clear way to assess financial impact.  What limits its implementation
is a lack of regulation requiring all counterparties to be revealed to a central bank or other
oversight agency.
The actual computation of the measures below can be demanding in practice, but it is feasible to
accurately approximate them, especially once a monitoring system is in place and constantly updated.

Let the {\sl net financial centrality} of $i$, given a
network $\mathbf{D}$ and vector of investments $\mathbf{q}$,
from a change
in its investment vector from $\mathbf{q}_i$ to $\mathbf{q}_i'$ be\footnote{Implicit in defining financial centrality, one has to take a stand on which equilibrium set of values is being used since those define the values $\mathbf{V}(\mathbf{q},\mathbf{p},\mathbf{D})$ and thus the bankruptcy costs $b(\mathbf{V}(\mathbf{q},\mathbf{p},\mathbf{D}),\mathbf{q}, \mathbf{p},\mathbf{D})$. Typically we are interested in either the best or worst equilibrium, but one could make other choices, or change from best to worst if one anticipates a freezing of payments in response to the failure of some organization(s) (e.g., see \cite{jacksonp2020}).  For some discussion about the importance of uncertainty about which equilibrium applies, see D'Errico and Roukny \citeyearpar{rouknybs2018}, and for strategic choices of defaults see Allouch and Jalloul \citeyearpar{allouch2017strategic}.}

\[
NFC_i(\mathbf{q},\mathbf{q}_i';\mathbf{D})=\mathbb{E}_{\mathbf{p}} \left[
\sum_j b_j(\mathbf{V}(\mathbf{q},\mathbf{p},\mathbf{D}),\mathbf{q},\mathbf{p},\mathbf{D}) - b_j (\mathbf{V}((\mathbf{q}_i', \mathbf{q}_{-i}),\mathbf{p},\mathbf{D}),(\mathbf{q}_i', \mathbf{q}_{-i}), \mathbf{p},\mathbf{D}) \right],
\]
This is the total impact on the economy  that comes
from a change in $i$'s investment strategy, based on the bankruptcy costs that are incurred. If there are no changes in bankruptcy costs, then the net financial centrality of $i$ is 0. This is {\sl net} since it is the impact beyond the direct change in $i$'s portfolio value.

This measure captures chains of cascading defaults, and so summarizes everything that is relevant in the network to be able to assess systemic risk. As a result, it can be a fairly complex function of the network structure $\mathbf{D}$, chosen portfolios $\mathbf{q}$, and distribution over asset returns $\mathbf{p}$. It often does not coincide with any standard measures of centrality, and we see this as a feature of real-world financial networks more than a bug, since financial contagion is a form of ``complex'' contagion.

The following examples illustrate the subtleties behind the notion of net financial centrality.

First consider a star network: Bank 1 owes $D_{i1}=D$ to every other Bank $i>1$ and similarly, every Bank $i>1$ owes $D_{1i}=D$ to Bank 1.
Furthermore, all banks owe $D_{0i}=1$ to outside depositors.
There are two assets as in Section \ref{core}: each bank has a safe asset that always pays $1+r\geq 1$ and a risky asset that pays $R$ with probability $\theta$ and zero otherwise. The correlation between banks' risky assets is governed by the parameter $m$ (the number of banks that get 0 realizations on the risky asset at the same time), as described in Section \ref{core}.
For simplicity, let $a$ and $\chi$ be large enough so that a defaulting bank repays nothing to its creditors.
Given the interbank star network $\mathbf{D}$, any standard network centrality measure would assign Bank 1 the highest centrality.
Consider what happens when $m=n$ such that risky assets are perfectly correlated across banks, for example all investing in the same subprime mortgages.
If banks all invest in the risky asset, all banks default with probability $1-\theta$. What happens if we changed Bank 1's portfolio to be only composed of the safe asset? If $1+r<1+(n-1)D$, nothing changes: in the bad state of the world, all of Bank 1's counterparties default on their debts to 1, which is enough to bring Bank 1 to bankruptcy. Bank 1's net financial centrality is then 0. However, changing a peripheral bank's portfolio to the safe asset strictly reduces total expected bankruptcy costs if that bank can sustain the default of its only counterparty (i.e., if $1+r\geq 1+D$). In that case, each Bank $i>1$ has strictly greater net financial centrality than Bank 1.
Of course, other parameterizations for portfolio returns and outside debts, would lead to other centrality rankings, some in line with Bank 1 being the most central bank.

This is not particular to correlated asset values, but depends more generally on the details of interdependencies and investments.
To see this, consider a different variation on a star-network.  Suppose that Bank 1 owes $D$ to each of Banks 2 through $n$, with $n\geq 4$ (who owe nothing back to Bank 1).
In addition, Bank 2 owes 2$D$ to Bank 3,  Bank 3 owes 2$D$ to Bank 4,  and so on,  with  Bank $n$ owing 2$D$ to Bank 2.  So, the peripheral banks form a sort of directed wheel. Furthermore, every bank owes 1 to outsiders.
Again, let $a$ and $\chi$ be large enough so that a defaulting bank repays nothing to its creditors.
Bank 1 has a larger total debt to more banks than any other bank, and again would be more central by standard network measures that ignore defaults and bankruptcy.
Suppose that the risky assets are independent.   Bank 1's risky asset pays $R_1>(n-1)D + 1 $ with probability $\theta$ and 0 with probability $1-\theta$,  and other banks' risky assets pay $D+1 > R_{-1}>1$ with probability $\theta$ and 0 with probability $1-\theta$.   Thus, Bank 1 also has more at risk than other banks.
Note, however, that Bank 1's default will not trigger a default by any other bank, unless the bank is already defaulting for other reasons.
Despite Bank 1 having the most total debt and to the most banks, no other bank has enough exposure to Bank 1 for it to trigger any form of cascade.  In contrast, the other banks have more concentrated exposures to each other so that any single one of their defaults triggers default by all other peripheral banks even if they all have positive returns.
The key in this example is that even though the peripheral banks have much less debt and liabilities, those are more concentrated and can be pivotal and cascade because of the pattern of exposures.
By changing Bank 1's investment to the risk-free asset, its own bankruptcy costs are saved but no systemic costs.   By changing any other bank's investment to the risk-free asset, there is a probability that the bankruptcy costs of all banks $2$ through $n$ are saved.   For a variety of parameters, the net financial centralities of the peripheral banks are (much) larger than Bank 1's.

What these examples show is that net financial centrality depends not only on the base network structure, but also upon how the pattern of debts and investments relate to the network structure - it embodies exactly what is needed to determine the potential impact of changing any given bank's portfolio.

More generally, if enough structure is imposed on the investments and debts, then net financial centrality can simplify and align with existing centrality measures. For instance, net financial centrality aligns with the threat index in Demange \citeyearpar{demange2016} if there is no uncertainty about asset returns, at realized portfolio returns $\mathbf{qp}$ all banks default, and bankruptcy costs consist solely of some share of a bank's assets (i.e., $\chi=0$).\footnote{Demange \citeyearpar{demange2016} does not consider bankruptcy costs, so her model corresponds to $a=0$, but her analysis can easily be extended to accommodate for $a>0$.} How changes in a bank's portfolio affects how much all banks repay each other, and the associated bankruptcy costs, can then be written as an infinite sum of walks in the interbank network, which is reminiscent of Katz-Bonacich centrality.

\subsection{Macroprudential Regulation versus Bailouts}\label{reg}

Consider a regulator who is deciding on the best way to regulate Bank $i$'s investments,
taking as given the financial network $\mathbf{D}$
and other banks' (possibly regulated) equilibrium  investments $ \mathbf{q}_{-i}$.
Focusing on a single bank while taking the remaining network of
banks' investments as given provides necessary conditions that optimal regulation must satisfy on the margin as well as insights into
regulatory tradeoffs.
Optimal regulation of a whole network involves interacting incentives, as we saw for core-periphery networks in Section \ref{sec:coreperiphery}.

In this analysis, we take into account the bank's optimum investment in response to the regulation.

Suppose the regulator can intervene in two ways to reduce the inefficiency of  banks'
investments: it can use a macroprudential policy that restricts the types of investments banks can make, or it can bailout insolvent banks ex post at some cost. 
Macroprudential policies and bailouts induce different investment decisions, and hence different returns on banks' investments as well as different default risks. 
We examine that tradeoff to compare ex post bailouts to policies of
portfolio regulation, as well as to laissez-faire where banks are left to invest as they
choose and not bailed out if they default.

We consider the setting of Section \ref{overly} in which a Bank $i$ has access to two assets, a risky asset and a safe asset. Without regulation, it would
fully invest in the risky asset.
As before, a macroprudential policy takes the form of an  upper bound $\bar{q}_{ir}$ on the share of its portfolio that $i$ can invest in the risky asset.
In this model, that could equivalently be imposed as either a reserve requirement of a minimum amount
of risk-free asset to be held, an upper bound on the variance of a bank's portfolio,
or a limit on the amount of the portfolio that is at risk.

Proposition \ref{risky} implies that a restriction on the amount a bank can invest in the risky asset
will be binding, and the bank best-responds to regulation by investing the maximum it can in the risky asset.
The optimal restriction balances the gain from the risk premium with
the \emph{overall societal} expected bankruptcy costs,
and solves
\[\max_{q_{ir}}\; q_{ir}\mathbb{E}[p_r]+ (1-q_{ir})(1+r) - \mathbb{E}_{\mathbf{p}}\left[\sum_{j} b_j(\mathbf{V}((1-q_{ir},q_{ir}),  \mathbf{q}_{-i},\mathbf{p},\mathbf{D}),\mathbf{q}, \mathbf{p},\mathbf{D})\right].\]
The first part of the objective function captures the expected value of the portfolio, which is linearly increasing in $q_{ir}$ because of the risk premium. The second term, capturing expected bankruptcy costs, jumps down discontinuously at some values of $q_{ir}$ where marginally changing $i$'s investment changes the set of defaulting banks in some states of the world. Because of these discontinuities, there are several levels of investment that can be optimal and need to be compared to each other. A natural one is to not regulate Bank $i$, and to let it choose $q_{ir}=1$.  Another one is the critical level of investment in the risky asset $\bar{q}_{ir}$ under which Bank $i$ remains solvent irrespective of the realization of the risky asset. This threshold solves
\[
(1-\bar{q}_{ir})(1+r) = D_i^L  \ \ \ \  {\rm or} \ \ \ \     \bar{q}_{ir} = 1-\frac{D_i^L}{1+r}.
\]

Letting $\mathbf{\bar{q}_{i}} =(1-\bar{q}_{ir},\bar{q}_{ir})$, it then follows that imposing a cap on risky investment of $\bar{q}_{ir}$ (equivalently, a reserve requirement of $1-\bar{q}_{ir}$) on Bank $i$ improves on laissez-faire if and only if

\begin{align*}
    \mathbb{E}[p_r] - &\mathbb{E}\left[\sum_{j} b_j(\mathbf{V}((0,1),  \mathbf{q}_{-i},\mathbf{p},\mathbf{D}),((0,1), \mathbf{q}_{-i}), \mathbf{p},\mathbf{D})\right]\\[5pt]
&< \left( 1-\frac{D_i^L}{1+r}\right) \mathbb{E}[p_r] +  \left(\frac{D_i^L}{1+r}\right) (1+r) - \mathbb{E}\left[\sum_{j} b_j(\mathbf{V}(\mathbf{\bar{q}_{i}},  \mathbf{q}_{-i},\mathbf{p},\mathbf{D}), (\mathbf{\bar{q}_{i}}, \mathbf{q}_{-i}), \mathbf{p},\mathbf{D})\right],
\end{align*}
which simplifies to
\begin{equation}
\label{reserveR}
\left[\frac{\mathbb{E}[p_r]}{1+r} - 1\right]D_i^L<
NFC_i((0,1), \mathbf{q}_{-i},\mathbf{\bar{q}_{i}}; \mathbf{D})
\end{equation}
Thus, preventing $i$'s default via restrictions on its portfolio is beneficial if the opportunity cost of doing so --
the loss in risk premium needed to ensure $i$'s solvency --
is less than the expected reduction in bankruptcy costs.
A regulator maximizing social welfare will therefore restrict the portfolio
when the risk premium is below a threshold, or alternatively when a bank's net financial centrality
is above a threshold. Hence, optimally, macroprudential regulation should be discretionary: tailored to the returns the bank could achieve as well as the centrality of its position in the network.

We remark that the opportunity cost of restricting investments (the left hand side of the above inequality) is increasing in the bank's outstanding debt.
As debt increases, a greater investment in the safe asset is required to avoid default,
which is increasingly costly in terms of expected returns to investment.
However, the net financial centrality of $i$ can also be
increasing in $D_i^L$, and so whether reserve requirements are
less likely to be optimal as outstanding debt increases
can depend on details of the
network structure and bankruptcy costs.

Next, suppose that the regulator also has the possibility of bailing out Bank $i$
when it is insolvent,
but this involves some expected cost $c_i>0$.\footnote{This cost not only captures the net capital injection the regulator makes,  but also any additional indirect
costs (not including costs that are recoverable). Lucas \citeyearpar{lucas2019measuring} estimates that the (unrecovered) direct cost of the 2008 crisis-related
bailouts in the United States was around \$500 billion.}
Bailing out Bank $i$ improves over laissez-faire if $i$'s bailout cost is lower than the total expected bankruptcy costs it would trigger by defaulting: $c_i<NFC_i((0,1), \mathbf{q}_{-i},\mathbf{\bar{q}_{i}}; \mathbf{D})$.

The comparison between bailing out a bank and restricting its investments is more subtle, as both prevent a bank's default but at different (opportunity) costs for society. A bailout is preferred to a reserve requirement of $1-\bar{q}_{ir}$ whenever
\begin{align*}
\left[\frac{\mathbb{E}[p_r]}{1+r} - 1\right]D_i^L> &\underbrace{\mathbb{E}_{\mathbf{p}}\left[\sum_{j} b_j((V_i((0,1), \mathbf{q}_{-i},  \mathbf{p},\mathbf{D})^+, \mathbf{V}_{-i}((0,1), \mathbf{q}_{-i},  \mathbf{p},\mathbf{D})),(0,1), \mathbf{q}_{-i}, \mathbf{p},\mathbf{D})\right]}_{\text{bankruptcy costs when $i$ invests freely to $\mathbf{q}_i=(0,1)$ and is bailed out when insolvent}}\\[5pt]
 &\qquad\qquad\qquad\qquad\qquad\qquad-\underbrace{ \mathbb{E}_{\mathbf{p}}\left[\sum_{j} b_j(\mathbf{V}(\mathbf{\bar{q}_{i}}, \mathbf{q}_{-i},  \mathbf{p},\mathbf{D}), \mathbf{\bar{q}_{i}}, \mathbf{q}_{-i},  \mathbf{p},\mathbf{D}) )\right]}_{\text{bankruptcy costs when $i$'s portfolio is regulated  }}+c_i
\end{align*}
which is equivalent to
\begin{equation}
\label{bailoutC}
\left[\frac{\mathbb{E}[p_r]}{1+r} - 1\right]D_i^L 
> c_i.
\end{equation}

Restricting a portfolio or providing bailouts each imply that a bank never defaults, which can have large effects on social surplus if, for instance,
the bank's solvency triggers a repayment cascade.
Net financial centrality captures how preventing $i$ from defaulting affects the
overall expected bankruptcy costs.

The above inequalities characterize the optimal regulation.  
We see that bailouts are optimal when both the expected excess returns from the risky investment and Bank $i$'s centrality are high.
A large risk premium means that one wants to take advantage of the risky asset's returns, while
a high centrality implies that bailouts are preferred to laissez-faire.
If expected returns are lower, but centrality is still high, then imposing portfolio restrictions are preferable to laissez-faire.
In contrast, if centrality is low enough, but expected returns are high, then laissez-faire becomes optimal.
In short:
\begin{itemize}
\item laissez-faire is best when centrality is low and excess returns are relatively high;
\item portfolio restrictions become optimal when excess returns are low and centrality is relatively high; and
\item bailouts are optimal when excess returns are high and so is centrality.
\end{itemize}

This is what was illustrated in Figure \ref{figure-regulate} in the introduction, which depicts how optimal regulation depends on a bank's expected excess returns from its available investment opportunities and the bank's centrality, holding bailout costs constant. 
Since there are three variables in question (expected excess return, centrality, and bailout cost), we also offer a different depiction in the right-hand panel of Figure \ref{figure-regulate_cs}, which depicts how optimal regulation depends on the regulator's baliout cost and a bank's centrality, holding expected excess returns constant.   

\begin{figure}[H]
\centering
\includegraphics[width=0.45\textwidth]{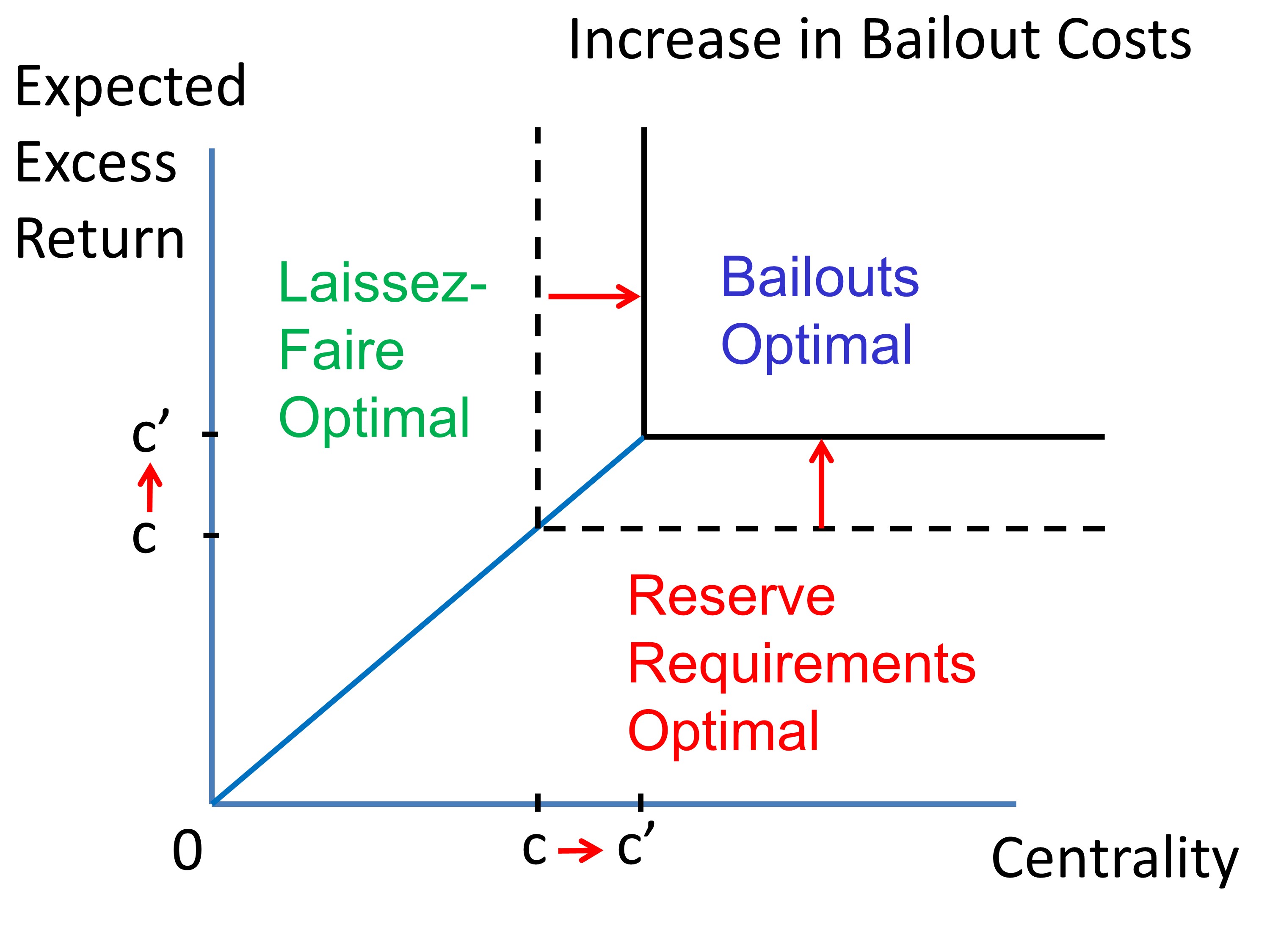}
\includegraphics[width=0.05\textwidth]{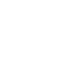}
\includegraphics[width=.45\linewidth]{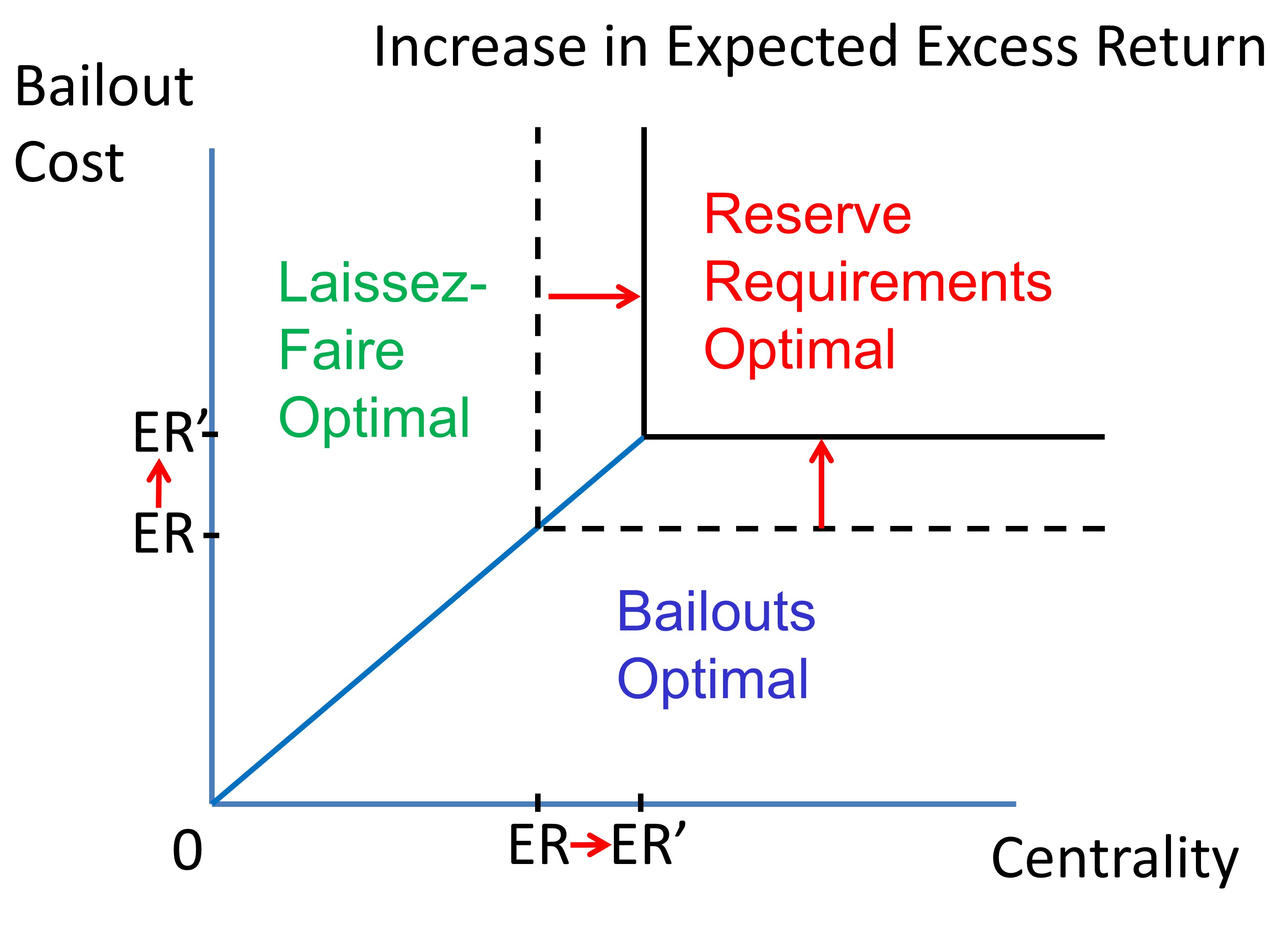}

\caption{\label{figure-regulate_cs}
The optimal regulation of a bank as a function of the expected excess return of the bank's available risky investments, the centrality of the bank, and the bailout cost.
The level of the reserve requirement is $1-\bar{q}_{ir}$.
The first panel pictures the optimal regulation as a function of the bank's expected excess return and centrality and shows how it changes due to an increase in the regulator's bailout costs;  while the second panel pictures the optimal regulation as a function of the baliout cost and centrality and shows the change in optimal regulation due to  an increase in the bank's expected excess return.
}
\end{figure}

Figure \ref{figure-regulate_cs} also depicts how optimal regulation 
changes as a function of bailout costs and expected excess returns.  Increasing bailout costs makes bailouts less attractive, while increasing expected excess returns makes reserve requirements less attractive.

\section{Discussion}

\subsection{Exogeneity of Interbank Exposures}\label{sec:endonetwork}

A limitation of our analysis is that it takes the network of interbank exposures $\mathbf{D}$ as given. This is relevant if these exposures capture long-lasting relationships between financial institutions that do not respond significantly to changes in their outside investments. For instance, the cross-ownerships that tie together banks in the French Cr\'edit Agricole Group are persistent and effectively constant.
In some settings, however, the amount that banks are willing to lend to each other responds to their portfolios and risk profiles. If a bank takes on excessively risky investments and may fail on its obligations, then others may ask for greater interest rates or not lend to it altogether.\footnote{For some evidence of this in short-term lending, see \cite{afonso2011stressed,cassola20132007}.}
Such considerations could then influence a bank's decision of whether to take on a risky investment, thus providing some market discipline and attenuating inefficiencies.
We consider this in the following example.
In particular, we show that it is not enough to induce efficient investments because of the cascade externalities that extend beyond a bank's partners.

\begin{example}
The example centers around Bank 1, which decides how to allocate its portfolio between a safe asset and a risky asset that pays $R$ with probability $\theta$ and nothing otherwise. The main difference with our above analysis is that Bank 1 is not exogenously given one unit of capital to invest, but instead raises capital from outside depositors (nodes 0 and 0'). Bank 1 has direct access to some depositors (node 0, on the left in Figure \ref{fig:exendogenous}) but not to all depositors (node 0'); e.g., because the latter are located in a different geographic area. To raise capital from these other depositors, it does so indirectly through Bank 2 and Bank 3. Both depositors 0 and 0' have a unit of capital to invest. All banks have access to a safe asset but only Bank 1 has access to the risky asset, thus providing a rationale for interbank lending.\footnote{We focus on the case in which these deposits or loans to Bank 1 are not insured, as otherwise there would be no response by others to Bank 1's investment portfolio, and hence no possibility for market discipline.}
The example is designed to highlight the role of cascades of defaults in magnifying inefficiencies, and even though it is stylized it captures essential features of interbank markets.

\begin{figure}[!h]
\begin{center}
\begin{tikzpicture}[scale=1]
\definecolor{afblue}{rgb}{0.36, 0.54, 0.66};
\foreach \Point/\PointLabel in {(0,0)/0, (3,0)/1, (6,0)/2, (9,0)/3, (12,0)/0'}
\draw[fill=afblue!40] \Point circle (0.35) node {$\PointLabel$};
\draw[<-, thick] (0.5,0) to node[above] {$D_{01}$} (2.5,0);
\draw[->, thick] (3.5,0) to node[above] {$D_{21}$} (5.5,0);
\draw[->, thick] (6.5,0) to node[above] {$D_{32}$} (8.5,0);
\draw[->, thick] (9.5,0) to node[above] {$D_{0'3}$} (11.5,0);
  \end{tikzpicture}
  \end{center}
  \captionsetup{singlelinecheck=off}
  \caption[]{\small Bank 1 has liabilities to both depositors (node 0) and Bank 2. Bank 2 has liabilities to Bank 3, which has liabilities to (some other) depositors (node 0'). }\label{fig:exendogenous}
\end{figure}

Even though who might want to borrow from/lend to whom is exogenously fixed, the size of the deposits and debts are not. They are determined endogenously as they depend on how likely Bank 1 is to default, and hence on its equilibrium investments. Indeed, if Bank 1's counterparties expect it to make risky investments, then they will require a greater promised payment to still be willing to lend capital to Bank 1. Let $a=0$ and $\chi>0$ large enough such that if Bank 1 defaults, it does not make any partial payments.
Since depositors do not have investment opportunities of their own, they are willing to lend capital to banks as long as the expected net return on this investment, net of bankruptcy costs, is greater than zero, i.e., $\mathbb{E}[d_{01}(\mathbf{V}) - \chi\mathbbm{1}\{V_1<0\}]\geq 1$ and $\mathbb{E}[d_{0'3}(\mathbf{V}) - \chi\mathbbm{1}\{V_3<0\}]\geq 1$.\footnote{Inefficiencies would be even larger if outside depositors did not internalize bankruptcy costs. }\footnote{Technically, interbank repayments and bankruptcy costs are also functions of investments and realized asset values, but we keep this dependence implicit to make notation less cluttered.    }
Bank 2 on its own cannot raise capital, so it is willing to borrow capital from Bank 3 and lend it to Bank 1 if $\mathbb{E}[V_2^+] =\mathbb{E}[(d_{21}(\mathbf{V}) -d_{32}(\mathbf{V}))^+]\geq 0$.
Bank 3's outside option is to invest its depositors' capital into the risk-free asset, in which case its expected equity value is simply $r$. It is willing to lend capital to Bank 2 instead if $\mathbb{E}[(d_{32}(\mathbf{V}) -d_{0'3}(\mathbf{V}))^+]\geq r$.

There are only two portfolios that can be optimal for Bank 1: either it invests just enough in the risk-free asset so as to always be solvent (i.e., set $q_i$ such that $(1-q_i)(1+r) = D_{01}+D_{21}$) or it invests fully in the risky asset.
Consider what happens in the first case.
By construction, Bank 1 never defaults and so the promised payments can be as small as $D_{01} = 1$,  $D_{21} = D_{32}=1+r$ and $D_{0'3}=1$. Bank 1's expected equity value is
\[\mathbb{E}[V_2^+] = \theta R \left[2-\frac{D_{01}+D_{21}}{1+r}\right]= \theta R\frac{r}{1+r}.\]
If instead Bank 1 invests fully in the risky asset, then it defaults with probability $1-\theta$, in which case Bank 2 also defaults on Bank 3 and Bank 3 on its own depositors. To compensate, depositors must be offered greater promised payments: $D_{01} =D_{0'3}= [1+(1-\theta)\chi]/\theta$. To be willing to lend the capital to Bank 2, Bank 3 must be offered $D_{32}$ such that $\theta[D_{32}-D_{0'3}]\geq r$, which can be as low as $D_{32}=1+r +(1-\theta)\chi$. To break even, Bank 2 must be offered at least $D_{21} =D_{32}=1+r +(1-\theta)\chi$. Overall, Bank 1's expected equity value is then
\[\mathbb{E}[V_2^+]  = \theta\left[2R - D_{01}+D_{21}\right]=2\theta R -2-r-2(1-\theta)\chi.\]
Bank 1 finds it optimal to invest fully in the risky asset whenever
\[2\theta R -2-r-2(1-\theta)\chi\geq  \theta R\frac{r}{1+r}, \]
which simplifies to
\[(2+r)\left(\frac{\theta R-(1+r)}{1+r}\right)\geq 2(1-\theta)\chi.\]
However, investing fully in the risky asset is socially efficient if and only if
\[(2+r)\left(\frac{\theta R-(1+r)}{1+r}\right)\geq 3(1-\theta)\chi.\]
Note the difference on the right-hand-side, which stems from the fact that Bank 2's bankruptcy cost is not accounted for by its counterparties when deciding on promised payments. Hence, even though Bank 1's risk-taking translates into higher interest rates on its liabilities to both of its counterparties, its incentives are still misaligned with social efficiency. Note that this inefficiency gets worse the longer the chain of interbank lending as well as the greater the bankruptcy costs. 
Overall, even if banks can adjust their exposures to each other in response to investment decisions, they do not account for the overall network externalities when doing so. \qed
\end{example}

The above example takes depositors to be risk-neutral. If they were risk averse, higher promised payments  would be required to compensate for the probability of default. That being said, Bank 2's bankruptcy costs would still not be accounted for, nor any additional network externalities, so the inefficiencies would persist. 
Moreover, disciplining becomes even more challenging if banks can only contract bilaterally and if observing a counterparty's investments is costly or impossible. 
The basic qualitative ideas underlying our results---i.e., that bilateral contracts between banks generate externalities, distort incentives, and must be addressed while keeping the network in mind---should extend.

A deeper analysis of such considerations requires modeling the endogenous formation and benefits from the network, which is beyond the scope of this paper but an important issue to merge with our analysis.\footnote{Distorted incentives are also to be expected in the formation of the network itself,
as shown in the recent literature on network formation in financial
settings (see, e.g., Gofman \citeyearpar{gofman2017},
Babus and Hu \citeyearpar{babush2017}, Farboodi \citeyearpar{farboodi2017},
Wang \citeyearpar{wang2017}, Erol \citeyearpar{erol2019}, and Blasques \citeyearpar{blasques2018}.)}

\subsection{Feedback between Private Incentives and Regulation}

Our analysis of optimal regulation allows banks to change their investments in response to the specific macroprudential and bailout policies chosen by the regulator. Indeed, in our model, banks choose their portfolios $(q_i)_i$ taking as given the regulation. At a very basic level, reserve requirements impose restrictions on the types of portfolios banks can choose. Similarly, bailouts can, in principle, also impact banks' investment decisions and lead them to take on even more risks. However, under most network structures, banks in equilibrium already choose the riskiest investments available, and so this margin is absent in our analysis of optimal regulation. It could be incorporated and would mostly make bailouts less appealing compared to macroprudential regulation and laissez-faire. It should not however change our main qualitative results or the optimality of asymmetric regulation.

As discussed in the previous section, we have also taken the network as given, but it can react to the regulation.  For instance, Erol \citeyearpar{erol2019} shows the anticipation of bailouts leads firms to create more links than is socially optimal, generating a network hazard. On the contrary, tight liquidity requirements can hinder interbank lending and cross-bank insurance (Erol and Ordonez \citeyearpar{erol2017network}).

Regulation can also be on the network itself.   For instance,  the use of Central Counterparty Clearing Houses (CCPs) can help mitigate some incentive and regulation issues.
CCPs can also monitor positions and impose margin requirements and mitigate some excessive risk-taking, although it is not clear
that they show appropriate concern over portfolio correlations.
Moreover, one has to worry about providing the CCPs with appropriate incentives and should be concerned about their extreme centrality and size.\footnote{See for instance, Duffie and Zhu \citeyearpar{duffiez2011}.}  Large government-sponsored enterprises that process huge
amounts of securities have an uneven history of success, especially if one examines Fannie Mae and Freddie Mac's failures in the 2008 crisis.
More generally, considering
regulations that change the network is an important topic for further research.

\subsection{Concluding Remarks}

Externalities in financial networks lead banks to take overly risky positions and to overly correlate their portfolio with those of their counterparties.\footnote{There could
be benefits to correlated portfolios that we overlook -- such as some gains to specialization and
an easier understanding of a counterparty's investment -- and that require a fuller investigation.}
The incentives of financial institutions to overly correlate portfolios have been largely ignored.  That can
lead banks to have higher centrality, as the situations in which they might default are ones in which the
system is then more naturally fragile. Excessive correlation can also exacerbate the multiple equilibria problem, the implications of firesales, and negative inferences, given that banks'
balance sheets become correlated.

As we have shown, macroprudential regulation
and/or bailouts can address excessive risk.
Historically, such intervention has
been sub-optimal due to the fact that it has rarely been tailored to banks' specific network positions, and is usually only applied to
a subset of financial organizations that involve externalities (e.g., it misses much of the shadow banking system). Moreover,
as we have shown, optimal regulation can involve asymmetries that policymakers might not anticipate, even though those asymmetries can substantially improve societal welfare.

In this paper we have taken the set of ``banks'' in question as given.  As we have seen in the past decade, this set is a constantly moving target.  A dramatic increase in the amount of intermediation done by non-bank financial intermediaries (NBFIs, a form of ``shadow banking'') has shifted a sizeable amount of the intermediation from regulated to non-regulated institutions (e.g., \cite{acharya2024banks,fsb2025}).  Given that less regulated institutions can sometimes offer higher rates of return, and may be presumed by investors to be implicitly insured by a government, they can divert funds from well-supervised and relatively safer sectors to ones that become riper for systemic disruptions.  
This presents a challenge for a regulator who can only regulate part of the network, and also presents challenges for governments that have to constantly redefine the boundaries of regulation.

Regardless of the precise policy that one undertakes, developing and maintaining
a more detailed picture of the network and of the portfolios of banks and other financial institutions is a necessary
first step both to improving crisis management and
to better understanding and monitoring incentive distortions.

\smallskip

\begin{spacing}{1.0}
\bibliographystyle{ecta}
\bibliography{financeNetworks}



\begin{appendices}

\section{Proofs} \label{appendixA}

\noindent{\bf Proof of Proposition \ref{risky}:}
Fix any investment strategy of the other banks in the network $ \mathbf{q_{-i}}$. Let $\mu$ be the measure on  $\mathbf{p}$, the vector of all portfolio gross returns, and let
\[
A( \mathbf{q_i}) = \{ \mathbf{p} \ | \   q_{is}(1+r) + q_{ir}p_{ir}  + d_i^{A}(\mathbf{V}(\mathbf{q_i},\mathbf{q_{-i}},\mathbf{p} ,\mathbf{D} ), \mathbf{D})>
D_i^L \}.
\]
Note that $\mu(A((0,1)))>0$ by Chebychev's inequality since $\mathbf{p}$ is bounded,
$\mathbb{E}[  p_{ir} ]> D_i^L$, and all other variables are nonnegative.
This implies $\mu(A(\mathbf{q_i}))>0$ for any possible optimizing level of $\mathbf{q_i}$.

Consider any $\mathbf{q_i}\neq (0,1)$ for which $\mu(A(\mathbf{q_i}))>0$, and
let us examine the gain in utility that results from  increasing $q_{ir}$ to $q_{ir}+\varepsilon$.
We show that for any such $q_{ir}$ there is an $\varepsilon>0$ for which this gain is strictly positive, which then implies the only optimizer is $q_{ir}=1$.

Let $\mathbf{q_{i}'}=(q_{is}-\varepsilon,q_{ir}+\varepsilon)$. Note that
\begin{align*}
&\int_{A(\mathbf{q_{i}'})} V_i(\mathbf{q_{i}'}, \mathbf{q_{-i}},\mathbf{p} ,\mathbf{D}) d\mu(\mathbf{p})  -  \int_{A(\mathbf{q_{i}})} V_i(\mathbf{q_{i}}, \mathbf{q_{-i}},\mathbf{p} ,\mathbf{D}) d\mu(\mathbf{p})
  \\[5pt]
  &  \geq
  \int_{A(\mathbf{q_{i}'})\cap A(\mathbf{q_{i}})} \left[V_i(\mathbf{q_{i}'}, \mathbf{q_{-i}},\mathbf{p} ,\mathbf{D}) - V_i(\mathbf{q_{i}},\mathbf{q_{-i}},\mathbf{p} ,\mathbf{D}) \right] d\mu(\mathbf{p}) - \int_{A(\mathbf{q_{i}})\setminus A(\mathbf{q_{i}'}) } V_i(\mathbf{q_{i}}, \mathbf{q_{-i}},\mathbf{p} ,\mathbf{D}) d\mu(\mathbf{p}).
\end{align*}
Next, it must be that $d_i^{A}(\mathbf{V}(\mathbf{q_{i}'}, \mathbf{q_{-i}},\mathbf{p} ,\mathbf{D}), \mathbf{D}) = d_i^{A}(\mathbf{V}(\mathbf{q_{i}}, \mathbf{q_{-i}},\mathbf{p} ,\mathbf{D}), \mathbf{D})$ for all $\mathbf{p}\in A(\mathbf{q_{i}'})\cap A(\mathbf{q_{i}})$. Indeed, for all $\mathbf{p}\in A(\mathbf{q_{i}'})\cap A(\mathbf{q_{i}})$, Bank $i$ is solvent under both investment strategies $\mathbf{q_{i}}$ and $\mathbf{q_{i}'}$. It repays its debts in full in either case, and so it cannot be that what it receives from its counterparties depends on how much it invested in the risky asset. Then $V_i(\mathbf{q_{i}},\mathbf{q_{-i}},\mathbf{p} ,\mathbf{D}) - V_i(\mathbf{q_{i}}, \mathbf{q_{-i}},\mathbf{p} ,\mathbf{D}) = \varepsilon \left( p_{ir} -(1+r)\right)$ for all   $\mathbf{p}\in A(\mathbf{q_{i}'})\cap A(\mathbf{q_{i}})$. 

Furthermore, for $\mathbf{p}\in A(\mathbf{q_{i}})\setminus A(\mathbf{q_{i}'})$, it must be that $V_i((q_{is}-\varepsilon',q_{ir}+\varepsilon'), \mathbf{q_{-i}},\mathbf{p} ,\mathbf{D}) = 0$ for some $\varepsilon'<\varepsilon$ and that
$V_i((q_{is}-\varepsilon'',q_{ir}+\varepsilon''), \mathbf{q_{-i}},\mathbf{p} ,\mathbf{D}) > 0$, for all $\varepsilon'' \in [0,\varepsilon')$.
Thus, for all  $\mathbf{p}\in A(\mathbf{q_{i}})\setminus A(\mathbf{q_{i}'})$
\[
V_i(\mathbf{q_{i}}, \mathbf{q_{-i}},\mathbf{p} ,\mathbf{D})\leq   \varepsilon' (1+r)  \leq \varepsilon (1+r) .
\]

Then
\begin{align*}
 \int_{A(\mathbf{q_{i}'})} V_i(\mathbf{q_{i}'}, \mathbf{q_{-i}},\mathbf{p} ,\mathbf{D}) &d\mu(\mathbf{p})  -  \int_{A(\mathbf{q_{i}})} V_i(\mathbf{q_{i}}, \mathbf{q_{-i}},\mathbf{p} ,\mathbf{D}) d\mu(\mathbf{p}) \\[5pt]
 & \geq
  \int_{A(\mathbf{q_{i}'})\cap A(\mathbf{q_{i}})} \varepsilon \left( p_{ir} -(1+r)\right) d\mu(\mathbf{p})
  -   \varepsilon (1+r) \int_{A(\mathbf{q_{i}})\setminus A(\mathbf{q_{i}'}) } d\mu(\mathbf{p}).
\end{align*}

\begin{claim}\label{c1}
If $\mu(A(\mathbf{q_{i}}))>0$, then
 $\mu(A(\mathbf{q_{i}}) \setminus A(\mathbf{q_{i}'}) ) \longrightarrow 0$  as $\varepsilon\longrightarrow 0$ while
  $\mu(A(\mathbf{q_{i}}) \cap A(\mathbf{q_{i}'}) ) \longrightarrow \mu(A(\mathbf{q_{i}}))$ as $\varepsilon\longrightarrow 0$.
\end{claim}

\begin{proof}[Proof of Claim \ref{c1}]
Recall that we keep equilibrium selection fixed as we vary a bank's investment $\mathbf{q_{i}}$. This implies that for all $\mathbf{p}\in A(\mathbf{q_{i}})\setminus A(\mathbf{q_{i}'}) $, Bank $i$ receives enough repayments from others to be solvent under $\mathbf{q_{i}}$, but that it defaults under $\mathbf{q_{i}'}$ even if we were to assume its interbank repayments are unchanged. That is, for all $\mathbf{p}\in A(\mathbf{q_{i}})\setminus A(\mathbf{q_{i}'}) $,
\begin{align*}
q_{ir}p_{ir}+q_{is}(1+r) + &d^A_i(\mathbf{V}(\mathbf{q_{i}},\mathbf{q_{-i}},\mathbf{p} ,\mathbf{D} ), \mathbf{D})-D_i^L\geq 0\\ &>q_{ir}p_{ir}+q_{is}(1+r) -\varepsilon(1+r-p_{ir})+ d^A_i(\mathbf{V}(\mathbf{q_{i}},\mathbf{q_{-i}},\mathbf{p} ,\mathbf{D} ), \mathbf{D})-D_i^L.
\end{align*}
Hence, for all $\mathbf{p}\in A(\mathbf{q_{i}})\setminus A(\mathbf{q_{i}'}) $,
\begin{align*}
\varepsilon(1+r-p_{ir})>q_{ir}p_{ir}+q_{is}(1+r) + d^A_i(\mathbf{V}(\mathbf{q_{i}},\mathbf{q_{-i}},\mathbf{p} ,\mathbf{D} ), \mathbf{D})-D_i^L&\geq 0.
\end{align*}
Importantly, $i$ is solvent at $\mathbf{q_{i}}$ for all $\mathbf{p}\in A(\mathbf{q_{i}})\setminus A(\mathbf{q_{i}'}) $. Hence it repays its debts in full and how much it receives from its counterparties is independent of its own portfolio: $d^A_i(\mathbf{V}(\mathbf{q_{i}},\mathbf{q_{i}'},\mathbf{p} ,\mathbf{D} ), \mathbf{D})=d^A_i(\mathbf{p}_{-i}, \mathbf{D})$. As $\varepsilon\longrightarrow 0$, the LHS converges to zero and $A(\mathbf{q_{i}})\setminus A(\mathbf{q_{i}'}) $ converges to the set of $\mathbf{p}$ such that $q_{ir}p_{ir}+q_{is}(1+r) +d^A_i(\mathbf{p}_{-i}, \mathbf{D})$ is equal to a constant. This set has measure zero since the price vector $\mathbf{p}$ has an atomless distribution:  $\mu(A(\mathbf{q_{i}}) \setminus A(\mathbf{q_{i}'}) ) \longrightarrow 0$.

We now show that $\mu(A(\mathbf{q_{i}}) \cap A(\mathbf{q_{i}'}) ) \longrightarrow \mu(A(\mathbf{q_{i}}))$ as
$\varepsilon\longrightarrow 0$. Since Bank $i$ is still solvent for $\mathbf{p}\in
A(\mathbf{q_{i}'})$, it pays back its debts in full under the alternative investment strategy and must get the same debt payments from others: $V_i(\mathbf{q_{i}'}, \mathbf{q_{-i}},\mathbf{p} ,\mathbf{D} ) = V_i( \mathbf{q_{i}}, \mathbf{q_{-i}},\mathbf{p} ,\mathbf{D} )+\varepsilon(p_{ir}-1-r) \longrightarrow  V_i(\mathbf{q_{i}},\mathbf{q_{-i}},\mathbf{p} ,\mathbf{D} )$ as $\varepsilon\longrightarrow 0$ and $\mu(A(\mathbf{q_{i}}) \cap A(\mathbf{q_{i}'}) ) \longrightarrow \mu(A(\mathbf{q_{i}}))$.

\end{proof}

Therefore,
for any $\delta>0$, for all small enough $\varepsilon$ the gain in utility is
  at least
  \[
\varepsilon
\left[ [\mu(A( \mathbf{q_{i}}) )-\delta]
\mathbb{E} \left[  p_{ir} -(1+r) \left| A(\mathbf{q_{i}})\right. \right]
-
\mu(A(\mathbf{q_{i}}) \setminus A(\mathbf{q_{i}'}) ) ] (1+r)\right],
\]
which is at least
 \[
\varepsilon
[\mu(A(\mathbf{q_{i}}) )-\delta]
\mathbb{E} \left[  p_{ir} -(1+r) \right]
- \varepsilon (1+r)\mu(A(\mathbf{q_{i}}) \setminus A(\mathbf{q_{i}'}) ),
\]
which is strictly positive for small enough $\delta$ and $\varepsilon$, establishing that $\mathbf{q_{i}}=(0,1)$ is a strictly dominant strategy.

We have left to prove that for sufficiently large bankruptcy costs, Bank $i$'s dominant investment strategy is socially inefficient. Recall that $1+r\geq D_i^L$, such that if Bank $i$ were to fully invest in the safe asset, it is guaranteed to remain solvent. Given some investment strategy for other banks $\mathbf{q_{-i}}$, total surplus if Bank $i$ sets $\mathbf{q_{i}}=(1,0)$ equals
\[1+r +\sum_{j\neq i }\left(\mathbb{E}[\mathbf{q_{j}}\mathbf{p}-b_j(\mathbf{V}(\mathbf{q_{i}}=(1,0),\mathbf{q_{-i}},\mathbf{p} ,\mathbf{D} ),\mathbf{p}, \mathbf{D})]\right).\]
If instead Bank $i$ invests fully in the risky asset, total surplus equals
\begin{align*}
\mathbb{E}(p_{ir})  -\mathbb{E}[b_i(\mathbf{V}(\mathbf{q_{i}}=(0,1),\mathbf{q_{-i}}, &\mathbf{p}, \mathbf{D}),\mathbf{p}, \mathbf{D})]\\[5pt]
&+ \sum_{j\neq i }\mathbb{E}\left[\mathbf{q_{j}}\mathbf{p}-b_j(\mathbf{V}(\mathbf{q_{i}}=(0,1),\mathbf{q_{-i}},\mathbf{p}, \mathbf{D}),\mathbf{p}, \mathbf{D})\right].
\end{align*}
Note that in all states of the world in wich $p_{ir}<D_i^L-D^A_i$, Bank $i$ must default, since it does not have enough to cover its liabilities even if it receives as much as possible from its debtors. If bankruptcy costs are bounded below by $\overline{\chi}>0$, total surplus is then bounded above by
\[\mathbb{E}(p_{ir})  -\overline{\chi} \Pr(p_r<D_i^L-D^A_i)+ \sum_{j\neq i }\mathbb{E}\left[ \mathbf{q_{j}}\mathbf{p}-b_j(\mathbf{V}(q_i=1,q_{-i};\mathbf{p}, \mathbf{D}),\mathbf{p}, \mathbf{D})\right].\]
A sufficient condition for $\mathbf{q_{i}}=(0,1)$ to be socially inefficient is if social surplus under $\mathbf{q_{i}}=(1,0)$ is greater. This is true if
\begin{align*}
1+r-&\mathbb{E}(p_{ir})  +\overline{\chi} \Pr(p_{ir}<D_i^L-D^A_i) \\
&- \sum_{j\neq i}(\mathbb{E}[b_j(\mathbf{V}(\mathbf{q_{i}}=(1,0),\mathbf{q_{-i}};\mathbf{p}, \mathbf{D}),\mathbf{p}, \mathbf{D})]-\mathbb{E}[b_j(\mathbf{V}(\mathbf{q_{i}}=(0,1),\mathbf{q_{-i}};\mathbf{p}, \mathbf{D}),\mathbf{p}, \mathbf{D})])>0.
\end{align*}
Note that $i$ is always solvent if it invests fully in the safe asset, and $d_{ji}(\mathbf{V}(\mathbf{q_{i}}=(1,0),\mathbf{q_{-i}};\mathbf{p}, \mathbf{D}), \mathbf{D})=D_{ji}$. Thus $b_j(\mathbf{V}(\mathbf{q_{i}}=(1,0),\mathbf{q_{-i}};\mathbf{p}, \mathbf{D}),\mathbf{p}, \mathbf{D})\neq b_j(\mathbf{V}(\mathbf{q_{i}}=(0,1),\mathbf{q_{-i}};\mathbf{p}, \mathbf{D}),\mathbf{p}, \mathbf{D})$ only if $i$ defaults on (some of) its payments under $\mathbf{q_{i}}=(0,1)$ at $\mathbf{p}$. Furthermore, since $d^A_j(\mathbf{V}, \mathbf{D})- \beta_j\left(\mathbf{V},\mathbf{p}, \mathbf{D}\right)$ is nondecreasing in $\mathbf{V}$, it must be that $ b_j(\mathbf{V}(\mathbf{q_{i}}=(1,0),\mathbf{q_{-i}};\mathbf{p}, \mathbf{D}),\mathbf{p}, \mathbf{D})- b_j(\mathbf{V}(\mathbf{q_{i}}=(0,1),\mathbf{q_{-i}};\mathbf{p}, \mathbf{D}),\mathbf{p}, \mathbf{D})\leq D_{j}^A$. Thus the last term is is bounded below by $-\sum_j  D_{j}^A$. For $\overline{\chi}$ high enough, the above inequality strictly holds and Bank $i$'s dominant investment strategy is socially inefficient.\eproof

\bigskip

\noindent{\bf Proof of Proposition \ref{correlation1}:}
Suppose there exists an equilibrium in which banks all choose different assets, such that their portfolio returns are independent. Take the point of view of some Bank $i$, and consider what happens if $i$ were to choose the same asset as some Bank $j$ instead. Doing so would not change the joint distribution between $i$'s portfolio return $\mathbf{q}_{i}\mathbf{p}$ and that of other banks $(\mathbf{q}_{l}\mathbf{p})_{l\neq i,j}$. For any realization of $(\mathbf{q}_{l}\mathbf{p})_{l\neq i,j}$, it would however move all probability mass from states in which $(p_i, p_j)=(\overline{R}, \underline{R})$ to states in which $(p_i, p_j)=(\overline{R}, \overline{R})$, and all mass from states in which $(p_i, p_j)=(\underline{R}, \underline{R})$ to states in which $(p_i, p_j)=(\underline{R}, \underline{R})$. This is then a strictly profitable deviation for Bank $i$ whenever
\begin{align*}
    \mathbb{E}_{\mathbf{q}_{-ij}\mathbf{p}}[V^+_i(\mathbf{q}_{i}&\mathbf{p}=\overline{R}, \mathbf{q}_{j}\mathbf{p}=\overline{R}, \mathbf{q}_{-ij}\mathbf{p})] -\mathbb{E}_{\mathbf{q}_{-ij}\mathbf{p}}[V^+_i(\mathbf{q}_{i}\mathbf{p}=\underline{R}, \mathbf{q}_{j}\mathbf{p}=\overline{R}, \mathbf{q}_{-ij}\mathbf{p})] \\&> \mathbb{E}_{\mathbf{q}_{-ij}\mathbf{p}}[V^+_i(\mathbf{q}_{i}\mathbf{p}=\overline{R},\mathbf{q}_{j}\mathbf{p}=\underline{R}, \mathbf{q}_{-ij}\mathbf{p})] -\mathbb{E}_{\mathbf{q}_{-ij}\mathbf{p}}[V^+_i(\mathbf{q}_{i}\mathbf{p}=\underline{R},\mathbf{q}_{j}\mathbf{p}=\underline{R}, \mathbf{q}_{-ij}\mathbf{p})].
\end{align*}
We now prove that  $\underline{R}\leq D_i^L-D_i^A<\overline{R}$ and $\underline{R}<D_j^L-D_j^A$ for some $i$, $j$ with $D_{ij}>0$ is a sufficient condition for this inequality to hold. We prove it holds weakly for each realization of $\mathbf{q}_{-ij}\mathbf{p}$, and strictly for some. Since $\underline{R}\leq D_i^L-D_i^A<\overline{R}$, Bank $i$ must either default or be on the verge of defaulting whenever $\mathbf{q}_{i}\mathbf{p}=\underline{R}$. Thus $V_i^+(\mathbf{q}_{i}\mathbf{p}=\underline{R}, \mathbf{q}_{-i}\mathbf{p})=0$ for all $\mathbf{q}_{-i}\mathbf{p}$, and the second term on both the LHS and RHS is zero. Since debt repayments are weakly increasing in asset returns, it must be that $V^+_i(\mathbf{q}_{i}\mathbf{p}=\overline{R}, \mathbf{q}_{j}\mathbf{p}=\overline{R}, \mathbf{q}_{-ij}\mathbf{p})\geq V^+_i(\mathbf{q}_{i}\mathbf{p}=\overline{R},\mathbf{q}_{j}\mathbf{p}=\underline{R}, \mathbf{q}_{-ij}\mathbf{p})$, and the inequality always holds weakly. We have left to show that it holds strictly for some realization of $\mathbf{q}_{-ij}\mathbf{p}$. Set $\mathbf{q}_{-ij}\mathbf{p} = \mathbf{\overline{R}}$. Since all banks are solvent when they receive a high return, $V^+_i(\mathbf{q}_{i}\mathbf{p}=\overline{R}, \mathbf{q}_{j}\mathbf{p}=\overline{R}, \mathbf{q}_{-ij}\mathbf{p}=\mathbf{\overline{R}}) = \overline{R}+D_i^A-D_i^L>0$. Either $i$ defaults or not when $\mathbf{q}_{i}\mathbf{p}=\overline{R}$, $\mathbf{q}_{j}\mathbf{p}=\underline{R}$, $\mathbf{q}_{-ij}\mathbf{p}=\mathbf{\overline{R}}$. If it does, then the RHS is zero and the inequality holds strictly. If it does not, then $V_i^+(\mathbf{q}_{i}\mathbf{p}=\overline{R}, \mathbf{q}_{j}\mathbf{p}=\underline{R}, \mathbf{q}_{-ij}\mathbf{p}=\mathbf{\overline{R}}) = \overline{R}+d_i^A(\mathbf{v})-D_i^L<\overline{R}+D_i^A-D_i^L$ since Bank $j$ must default when $\mathbf{q}_{j}\mathbf{p}=\underline{R}$ and so cannot pay back its debt to $i$ in full. Thus the inequality holds strictly for Bank $i$ and independent portfolios is not part of an equilibrium. \eproof

\bigskip

\noindent{\bf Proof of Proposition \ref{correlation2}:}
Next, we show that there exists an equilibrium in which all banks invest in the same asset, such that their portfolio returns are perfectly correlated.
Given that all other banks are symmetric under the conjectured equilibrium, we look at the problem faced
by $i$. Note that if $i$ decides not to perfectly correlate its portfolio
then there must be at least one state in which its portfolio pays off but none of the others does,
and at least one state in which its portfolio does not pay off but all the others do.
The incentive condition boils down to
\begin{align*}
V_i(\mathbf{q}_{i}\mathbf{p}=\overline{R}, \mathbf{q}_{-i}\mathbf{p}=\mathbf{\overline{R}})^+ + V_i(\mathbf{q}_{i}\mathbf{p}=\underline{R},& \mathbf{q}_{-i}\mathbf{p}=\mathbf{\underline{R}})^+ \\
&\geq V_i(\mathbf{q}_{i}\mathbf{p}=\underline{R}, \mathbf{q}_{-i}\mathbf{p}=\mathbf{\overline{R}})^+
+V_i(\mathbf{q}_{i}\mathbf{p}=\overline{R}, \mathbf{q}_{-i}\mathbf{p}=\mathbf{\underline{R}})^+\\
\iff V_i(\mathbf{q}_{i}\mathbf{p}=\overline{R}, \mathbf{q}_{-i}\mathbf{p}=\mathbf{\overline{R}})^+ -V_i(\mathbf{q}_{i}\mathbf{p}=&\underline{R}, \mathbf{q}_{-i}\mathbf{p}=\mathbf{\overline{R}})^+ \\
&\geq V_i(\mathbf{q}_{i}\mathbf{p}=\overline{R}, \mathbf{q}_{-i}\mathbf{p}=\mathbf{\underline{R}})^+
-V_i(\mathbf{q}_{i}\mathbf{p}=\underline{R}, \mathbf{q}_{-i}\mathbf{p}=\mathbf{\underline{R}})^+.
\end{align*}
Thus, if this condition holds for all $i$, then all banks choosing perfectly correlated portfolios is an equilibrium.

We now prove that any of the followings is a sufficient condition: either (i) $i$ is not part of a cycle, or (ii) $D_i^L\leq \underline{R}$, or (iii) $ D_i^L-D_i^A \geq \underline{R}$. If $i$ is not part of a cycle, then how much debt payments it receives is completely independent of its own solvency status, and then of its portfolio: $d_i^A(\mathbf{q}_{i}\mathbf{p}=\overline{R}, \mathbf{q}_{-i}\mathbf{p}) = d_i^A(\mathbf{q}_{i}\mathbf{p}=\underline{R}, \mathbf{q}_{-i}\mathbf{p})$ for all $\mathbf{q}_{-i}\mathbf{p}$. If $i$ never defaults, then both the LHS and RHS are equal to $\overline{R}-\underline{R}$, and the inequality holds weakly. The worst case for the inequality if when $i$ defaults when $(\mathbf{q}_{i}\mathbf{p}=\underline{R}, \mathbf{q}_{-i}\mathbf{p}=\mathbf{\underline{R}})$ but not when $(\mathbf{q}_{i}\mathbf{p}=\underline{R}, \mathbf{q}_{-i}\mathbf{p}=\mathbf{\overline{R}})$. The LHS of the inequality is then equal to $\overline{R}-\underline{R}$, while the RHS is $\overline{R}+d_i^A(\mathbf{q}_{-i}\mathbf{p}=\mathbf{\underline{R}})-D_i^L$. However, for $i$ to default when $(\mathbf{q}_{i}\mathbf{p}=\underline{R}, \mathbf{q}_{-i}\mathbf{p}=\mathbf{\underline{R}})$, it must be that $\underline{R}+d_i^A(\mathbf{q}_{-i}\mathbf{p}=\mathbf{\underline{R}})-D_i^L<0$, and so the inequality holds.

If $D_i^L\leq \underline{R}$ (condition (ii)), then $i$ is always solvent. If $i$ is always solvent, it always pays back its debts in full to others, and so its own portfolio return cannot impact how much debt payments it receives from others. As above, both the LHS and RHS are equal to $\overline{R}-\underline{R}$, and the inequality holds weakly.

Finally, if  $ D_i^L-D_i^A \geq \underline{R}$ (condition (iii)), then $i$ either defaults or has zero equity value whenever $\mathbf{q}_{i}\mathbf{p}=\mathbf{\underline{R}}$. Then $V_i^+(\mathbf{q}_{i}\mathbf{p}=\underline{R}, \mathbf{q}_{-i}\mathbf{p})=0$ for all $\mathbf{q}_{-i}\mathbf{p}$, and only the first term on the LHS and RHS remain. Since bank values and debt repayments are both weakly increasing in portfolio returns, the inequality holds as well. \eproof

\bigskip

\noindent{\bf Proof of Proposition \ref{correlation3}:} Towards a contradiction, suppose that there exists an equilibrium in which banks do not perfectly correlate their portfolios. Let $N_k$ be the set of banks who choose asset $k$ in equilibrium. By assumption, there does not exist an asset $k$ such that $N_k = N$. Let $\underline{k}=\arg\min_{k:N_k\neq \emptyset} |N_k|$ and take any Bank $i\in N_{\underline{k}}$. Consider the deviation for Bank $i$ under which it invests fully in asset $\overline{k}=\arg\max_{k:N_k\neq \emptyset} |N_k|$. By definition, $|N_{\overline{k}}|$ of $i$ counterparties invest in asset $\overline{k}$ while only $|N_{\underline{k}}|-1<|N_{\overline{k}}|$ of them invest in asset $\underline{k}$. Such deviation does not change the correlation with portfolios of banks not in $N_{\overline{k}}\cup N_{\underline{k}}$. Indeed, for each possible realization of $(\mathbf{q}_{j}\mathbf{p})_{j\notin N_{\overline{k}}\cup N_{\underline{k}}}$, this deviation shifts probability mass $\theta(1-\theta)$ from states in which $p_i = \overline{R}$, $(\mathbf{q}_{j}\mathbf{p})_{j\in N_{\overline{k}}} = \mathbf{\underline{R}}$, $(\mathbf{q}_{j}\mathbf{p})_{j\in N_{\underline{k}}} = \mathbf{\overline{R}}$, to states in which $p_i = \overline{R}$, $(\mathbf{q}_{j}\mathbf{p})_{j\in N_{\overline{k}}} = \mathbf{\overline{R}}$, $(\mathbf{q}_{j}\mathbf{p})_{j\in N_{\underline{k}}} = \mathbf{\underline{R}}$. Similarly, it shifts  probability mass $\theta(1-\theta)$ from states in which $p_i = \underline{R}$, $(\mathbf{q}_{j}\mathbf{p})_{j\in N_{\overline{k}}} = \mathbf{\overline{R}}$, $(\mathbf{q}_{j}\mathbf{p})_{j\in N_{\underline{k}}} = \mathbf{\underline{R}}$, to states in which $p_i = \underline{R}$, $(\mathbf{q}_{j}\mathbf{p})_{j\in N_{\overline{k}}} = \mathbf{\underline{R}}$, $(\mathbf{q}_{j}\mathbf{p})_{j\in N_{\underline{k}}} = \mathbf{\overline{R}}$. Doing so shifts high returns to states in which strictly more banks also receive a high return, and so it constitutes a strictly profitable deviation whenever the value of a high return is strictly increasing in the number of other banks getting a high return. \eproof

\bigskip

\noindent{\bf Proof that all banks invest fully in the risky asset in Section \ref{sec:coreperiphery}}: Note that, in the state of the world in which $p_i=R$ for all $i$, all banks must be solvent in the best equilibrium irrespective of their exact portfolio. Indeed, their portfolio must have a value of at least $1+r$, which is enough to cover their liabilities to peripheral banks, and is thus enough to ensure solvency of all in the best equilibrium. So the only thing that matters from the point of view of Bank $i$ is the number of its counterparties that default in the adverse state of the world, even when $i$ pays back its debt. Call this number $X$.

There are two cases. When $m\leq k^R$, a bank that gets $p_i=R$ never defaults. Investing fully in the risky asset yields
\[\left[1-\frac{n_c(1-\theta)}{m}\right](R-D_0) + \frac{n_c(1-\theta)}{m}\frac{n_c-m}{n_c}(R-D_0-XD) = \theta (R-D_0) -\frac{(n_c-m)(1-\theta)}{m}XD.\]
Investing in the safe asset cannot improve over this. There are again, two sub-cases. Either a bank fully invested in the safe asset can sustain $X$ of its counterparties defaulting, or it cannot. In the first case, the associated payoff is then $1+r-D_0-\frac{n_c(1-\theta)}{m}XD$, which is worse than above since $\theta R > 1+r$. In the second case, Bank $i$ defaults in the adverse state of the world, such that its overall expected payoff is $[1-\frac{n_c(1-\theta)}{m}][1+r-D_0]$. Again, this does worse than investing fully in the risky asset since $R>D_0+XD$.

The second case is when $m>k^R$, that is when there can be a default cascade in the bad state of the world. Hence it is now possible that $X>m$, i.e. that some banks with high return default. There are again two sub-cases: either getting a high return $R$ allows Bank $i$ to sustain $X$ defaults (if $X\leq k^R$) or it doesn't. If it does, then investing in the risky asset yields the same payoff as above, and the same arguments apply. If getting $R$ is not enough to ensure $i$'s solvency, then fully investing in the risky asset yields
\[\left[1-\frac{n_c(1-\theta)}{m}\right][R-D_0].\]
But then investing in the safe asset cannot allow a bank to sustain $X$ defaults, and hence investing in the safe asset yields a lower total expected payoff of $[1-\frac{n_c(1-\theta)}{m}][1+r-D_0]$. Overall, investing in the risky asset is always a best response. \eproof
\bigskip

\noindent{\bf Proof of Proposition \ref{symmetric}}: Recall that when $m\leq k^R$, there is no default cascade: in the best equilibrium, all banks that get $p_i=R$ always remain solvent. Hence the only banks that default under laissez-faire are those with $p_i=0$. Furthermore, these banks can remain solvent only if they have some amount invested in the safe asset, i.e., only if they are regulated. If all banks are regulated, then the optimal limit on risky investments is $\bar{q}=1-D_0/(1+r)$, as this guarantees their solvency at minimal cost in terms of lost risk premium. Now suppose that only a subset of banks are regulated, and that $X$ banks remain unregulated. Then only a stricter restriction of $\bar{q}_X=1-(D_0+XD)/(1+r)$ on regulated banks guarantee their solvency when $p_i=0$.  But if such restriction improves over laissez-faire (i.e. if preventing the default of these banks when $p_i=0$ is worth imposing $\bar{q}_X$ on them) then only imposing $\bar{q}$ but on all banks yields an even greater social surplus as expected bankruptcy costs scale linearly with the number of unregulated banks. Hence the only macroprudential policy that can be optimal is symmetric, and imposes a reserve of $D_0/(1+r)$ on all banks.

This improves over laissez-faire if and only if
\begin{align*}
n_c\left[\theta R - \frac{D_0[\theta R - (1+r)]}{1+r}\right]>n_c[\theta R - (1-\theta)\chi]\iff \theta <\frac{\chi+D_0}{\chi+\frac{R}{1+r}D_0}.
\end{align*}

\bigskip

\noindent{\bf Proof of Proposition \ref{asymmetric}}:
For now, suppose that some regulation is optimal.

Suppose that all banks are symmetrically regulated to hold at most $q$ of the risky asset, for some $q<1$.

First we argue that when $m$ banks get 0 returns on the risky portfolio, it must be that they do not default.
To see this, note that the remaining banks who got a positive return on their risky investments have a total portfolio value of less than $R$ (given that $q<1$).
So, if the $m$ banks default, then since $m>k^R$ the remaining banks will also default, and so there is full contagion.   Thus, the regulation cannot have been optimal since it limited total expected returns but
had no impact on the defaults.

Thus, banks that get 0 returns must be solvent when all other banks are solvent.   This means that they must be able to pay their periphery banks.
Thus,
\[
(1-q)(1+r)\geq D_0
\]
It then follows that $(1-q)=\frac{D_0}{1+r}$, as any lower bound simply lowers return
without affecting solvencies.
This establishes (i).

The total expected net return from  investments by  core banks under such regulation (accounting for what they owe their depositors) is
\begin{equation}
\label{asymmetric1}
n_c
\left[ (1-q)(1+r) + q \theta R  - D_0 \right]  = n_c
\left[\theta R-\frac{D_0[\theta R - (1+r)]}{1+r} -D_0 \right].
\end{equation}

Consider the following alternative regulation.
Note that $k^r= \lfloor \frac{1+r-D_0}{D}\rfloor \geq 1$ is the maximum number of couterparty defaults that a bank fully
invested in the risk-free asset can sustain without defaulting itself.
Allow $k^r$ banks to invest freely, and regulate the remaining $n_c-k^r$ to hold enough of the safe asset to be able to sustain $k^r$ defaults: $(1-q)(1+r) = D_0+k^rD \iff q=  1-(D_0+k^rD)/(1+r)$.  Under this, at most $k^r$ banks ever default together and so there is no contagion.

Note that $k^r$ is the maximum number of correlated core defaults that the economy can take without having contagion within the core, because $k^r<m$ and so at least
one of the remaining regulated/non-defaulting banks will have no return on its risky investments, and the most it could withstand would be if it is fully invested in the risk-free asset.

Under this alternative, the expected net return from the investments is
\begin{equation}
\label{asymmetric2}
k^r \left(\theta R  - (1-\theta) \chi\right) +\left( n_c-k^r \right)\left(\theta R - (D_0+k^r D )\frac{\theta R - (1+r)}{1+r}\right)
-
n_cD_0,
\end{equation}
where the last term accounts for what they owe depositors.

Note that (\ref{asymmetric1}) is less than (\ref{asymmetric2}) whenever
\begin{align*}
(1-\theta)\chi<\frac{\left[ n_c -\left( n_c -k^r \right)(D_0+k^r D)\right]}{k^r}\left(\frac{\theta R - (1+r)}{1+r}\right).
\end{align*}

This is satisfied whenever $n_c>\left( n_c -k^r \right)(D_0+k^r D)$ and
$\theta >\underline{\theta}$.   This establishes that, if these conditions are satisfied, then the optimal regulation cannot be symmetric,
which establishes (ii).

Finally, let us consider conditions under which some regulation is optimal.

We establish conditions under which symmetric regulation to $(1-q)=\frac{D_0}{1+r}$ improves over no regulation.

No regulation is worse than symmetric regulation if and only if:
\[
n_c\left[\theta R -  \chi  \frac{n_c}{m}(1-\theta)  -D_0\right]< n_c\left[ D_0 + \left(1-\frac{D_0}{1+r}\right)\theta R -D_0 \right],
\]
where the term $ \frac{n_c}{m}(1-\theta)$ on the left hand side is the expected probability that some banks default, in which case $m$ default and there is total contagion.
This simplifies to:
\[
\theta< \overline{\theta} = \frac{\frac{n_c}{m}\chi+D_0  }{\frac{n_c}{m}\chi+\frac{R}{1+r}D_0 }  <1 ,
\]
which is equivalent to (\ref{thetahigh}).
This completes the proof of the proposition.

Finally, we check that there exist parameter values under which optimal regulation must be asymmetric, which is the case when:
\begin{align*}
\overline{\theta}> \underline{\theta}&\iff
\frac{\frac{n_c}{m}\chi+D_0}{\frac{n_c}{m}\chi+\frac{R}{1+r}D_0}> \frac{ n_c-\left(n_c-k^r\right)(D_0+k^rD)+k^r\chi}{\left[n_c-\left(n_c-k^r\right)(D_0+k^rD)\right]\frac{R}{1+r} + k^r\chi}\\
&\iff \frac{n_c}{m}\left[n_c-\left(n_c-k^r\right)(D_0+k^rD)\right]\geq k^rD_0.
\end{align*}
\eproof

\bigskip

\noindent{\bf Proof of Proposition \ref{prop:optregulation_NS}}:
When $m\leq k $,  there is no contagion in the best equilibrium: only banks with return $p_i=0$ default while all others remain solvent. As argued in the proof of Proposition \ref{symmetric}, there are complementarities in ensuring the solvency of banks jointly as then the necessary reserve requirements are lower: if a bank's counterparties are regulated and remain solvent, then the safe investment that that bank needs to remain solvent itself is lower. So either all banks in a tier are regulated or none of them are.

Consider first the least restrictive regulation that ensures all core banks always remain solvent. This is achievable by setting a (potentially tier-dependent) cap $\overline{q}^\ell$ on the risky investments of banks in tier $\ell$, for each $\ell$, such that
\[(1-\overline{q}^\ell)(1+r) + (x^\ell-1)D = (x^\ell-1)D + D_0,\]
where $x^\ell$ is the number of core counterparties of banks in tier $\ell$. This simplifies to
\[  \overline{q}^\ell = 1-\frac{D_0}{1+r},\]
This improves over laissez-faire if and only if
\[\theta \overline{q}^\ell R+ (1-\overline{q}^\ell)(1+r)> \theta R - (1-\theta) \chi ,\]
which simplifies to
\[[\theta R - (1+r)]\frac{D_0}{1+r}<(1-\theta) \chi \iff \theta <\frac{\chi+D_0}{\chi+\frac{R}{1+r}D_0}.\]
Only regulating some tier(s) but not others cannot be optimal, as that would require a stricter restriction on investments of the regulated tier(s) for the same expected benefits. \eproof

\bigskip

\noindent{\bf Proof of Proposition \ref{prop:optregulation_NS2}}: As argued in the proofs of previous propositions, an optimal symmetric regulation either imposes no restriction on core banks or requires all of them to hold $1-\overline{q}=D_0/(1+r)$ in the risk-free asset. No bank ever default under the latter, and social surplus per bank equals
\[(1-\overline{q})(1+r) + \overline{q} \theta R   =
\theta R-\frac{D_0[\theta R - (1+r)]}{1+r}.\]
We build an alternative asymmetric regulation that imposes more restrictions on higher-tier banks. Note that the subnetwork of core banks has a nested-split structure, which means that it can be partitioned into two sets: one set of banks that form a clique and an independent set (i.e., a set of core banks that are only exposed to banks from the clique but not to each other). That is, there exists a threshold $\overline{l}$ such that banks in tiers $\overline{l}+1, \dots L$ are all exposed to each other and form a clique, whereas banks in tiers $1,\dots \overline{l}$ are only exposed to banks in tiers $l>\overline{l}$. Consider the following regulation. All banks in the clique are fully regulated: they invest only in the safe asset, and so can sustain up to $k^r$ defaults. For now, suppose that the independent set contains at least $k^r$ banks. Arbitrarily choose $k^r$ banks from that set and let them invest freely. Finally,  impose a reserve requirement of $D_0/(1+r)$ on the remaining banks from the independent set. Note that the latter are only exposed to the clique, and so if the clique remains solvent, so do they. The clique remains solvent if at most $k^r$ of its counterparties default. Hence only the $k^r$ that invest freely default when they get $p_i=0$, and the regulated clique insulates the rest of the network from these defaults.

Letting $N^{clique}\equiv \sum_{\ell>\overline{\ell}}n_c^\ell$ and $N^{ind}\equiv \sum_{\ell\leq\overline{\ell}}n_c^\ell$ denote the number of core banks in the clique and independent set, respectively, the above policy yields a total surplus of
\[N^{clique}(1+r) + (N^{ind}-k^r)\left[\theta R \frac{r}{1+r}+1\right]+k^r[\theta R - (1-\theta)\chi].\]
This is better than the above symmetric policy whenever
\begin{align*}
N^{clique}(1+r) + (N^{ind}-k^r)\left[\theta R \frac{(1+r-D_0)}{1+r}+D_0\right]&+k^r[\theta R - (1-\theta)\chi]
\\&>(N^{clique}+N^{ind})\left[\theta R \frac{(1+r-D_0)}{1+r}+D_0\right],
\end{align*}
which simplifies to
\[\theta >\frac{\chi+D_0-\frac{(1+r-D_0)}{k^r}N^{clique}}{\chi+\frac{R}{1+r}(D_0-\frac{(1+r-D_0)}{k^r}N^{clique})}.\]
  \eproof

\bigskip

\setcounter{page}{0}\thispagestyle{empty} \newpage

{\bf Supplemental (Online) Appendix for ``Optimal Regulation and Investment Incentives in
Financial Networks''  by Jackson and Pernoud}

\section{Additional Discussions} \label{appendixB}

\subsection{General Contracts Between Financial Organizations}
\label{generalf}

Accounting for both debt and equity contracts can be important since they can lead to different
incentives for investment and are both prevalent in practice, and thus in this section we examine that model and comment on extensions of the results of the paper to this case.\footnote{
Some banks' balance sheets are mostly debt-based, involving deposits,
loans, collateralized debt obligations, and other sorts of
debt-like instruments. Other organizations are however hybrids that involve substantial portions of both types of
exposures. For instance, some of the largest banks in France, including the French Cr\'edit
Agricole Group, consist of a network of numerous banks with significant cross-ownership between
each other, sometimes in excess of fifty percent. This also captures interactions within financial
conglomerates, as organizations
within the same group often have both debt-like and equity-like exposures
to each other (e.g., see the report of the Basel Committee on Banking Supervision \citeyearpar{BISreport_conglomerates}
for a description of intra-group exposures).
In addition, many syndicated investment portfolios, as well as venture capital firms and many other
sorts of investment funds, typically hold equity in other organizations.
}
Thus, here we show how the model and some of the results extend beyond debt interdependencies.

\subsubsection{Debt and Equity}

We begin with the extension to include equity interdependencies.
In addition to the debt holdings in the paper, if Bank $i$ owns an equity share in Bank $j$, it is represented as
a claim that $i$ has on a fraction $S_{ij}\in [0,1]$ of $j$'s value. Let $\mathbf{D}$ the matrix of debt claims and $\mathbf{S}$ the matrix of
equity claims.
A bank cannot have a debt or equity claim on itself, and so
$S_{ii}=D_{ii}=0$ for all $i$.
Equity shares sum to one, so whatever share is not owned by other
banks accrues to some outside investor:
$S_{0i}=1-\sum_{j\neq 0} S_{ji}$.\footnote{Here, we simply model
any fully
privately held banks as having some outside investor owning
an equity share equal to 1.
This has no consequence, but allows us to trace where
all values ultimately accrue.}
The one exception is that no shares are held in the outside investors so that $S_{i0}=0$ for all $i$ -- shares held by banks in private enterprises are modeled via the $p_i$'s.

Finally, in order to ensure that the economy is well-defined,
we presume that there exists a
directed equity path from every bank to some private investor (hence to node 0).
This rules out nonsensical cycles where each bank is {\sl entirely}
owned by others
in the cycle, but none are owned in any part by any private investor.
For instance
if $i$ owns all of $j$ and vice versa, then there is no solution to equity values.

\begin{equation}
\label{debtvalue2}
d_{ij}(\mathbf{V}) =   \frac{D_{ij}}{\sum_h D_{hj}} \max\left(\sum_k q_{jk} p_k+ d_{j}^A(\mathbf{V}) +\sum_{h}S_{jh} V_h^+ - \beta_j(\mathbf{V},\mathbf{p}),0\right),
\end{equation}
where $d_{j}^A(\mathbf{V})\equiv \sum_{h}d_{jh}(\mathbf{V})$ and $V_h^+\equiv\max \{V_h, 0\}$.

Equity holders of Bank $j$ do not receive any payment when $j$ is insolvent:
\[
S_{ij}V_j^+= 0.
\]
This reflects limited liability: the value to $i$ of its equity holding in $j$ cannot be negative.

A bank defaults whenever the value of its assets does not cover its liabilities. The bankruptcy costs it incurs are then
\begin{equation}
\label{bankrupt2}
  b_i\left(\mathbf{V},\mathbf{p}\right) =
  \begin{cases}
   0 & \text{ if  } \sum_k q_{ik} p_k +\sum_{j} S_{ij} V_j^+  +d_i^A(\mathbf{V})  \geq D_i^L \\
   \beta_i\left(\mathbf{V},\mathbf{p}\right) & \text{ if  } \sum_k q_{ik} p_k +\sum_{j} S_{ij} V_j^+  +d_i^A(\mathbf{V})  <D_i^L.
  \end{cases}
\end{equation}

The book value of a bank is:
$$
V_i=\sum_k q_{ik} p_k +\sum_{j} S_{ij} V_j^+  +   d_j^A(\mathbf{V})-D^L_i  - b_i(\mathbf{V},\mathbf{p}),
$$
where $b_i(\mathbf{V},\mathbf{p})$ is defined by (\ref{bankrupt2}) and $ d_j^A(\mathbf{V})$ by (\ref{debtvalue2}).
Banks' equilibrium values solve:
\begin{equation}
\label{eq-bookvalue-bankruptcy2}
\mathbf{V}=(\mathbf{I}-\mathbf{S}(\mathbf{V}))^{-1} \left(\left[ \mathbf{q} \mathbf{p} + \mathbf{d}^A(\mathbf{V})  - \mathbf{D}^L\right] - \mathbf{b}(\mathbf{V},\mathbf{p})\right),
\end{equation}
where  $\mathbf{S}(\mathbf{V})$ reflects the fact that $S_{ij}(\mathbf{V}) = 0 $ whenever $j$ defaults, and
equals $S_{ij}$ otherwise.

Note that the value of a bank is weakly increasing in that of others in the network, which ensures the existence of a solution to equation (\ref{eq-bookvalue-bankruptcy2}), as discussed above.   This is implied by Tarski's fixed point theorem, since bank values depend monotonically on each other and are bounded.\footnote{They are bounded above by $\bar{\textbf{V}} = (\textbf{I} - \textbf{S})^{-1}[\textbf{q}\textbf{p}+\textbf{D}^A]$.  (Those are unique and exist, since $(\mathbf{I}-\mathbf{S})^{-1}$ exists under our assumption that there are equity paths from all banks to node 0 and none going back.)  }

The case in which $S_{ij}={0}$ for all $i>0,j>0$ nests the models of Eisenberg and Noe
\citeyearpar{eisenbergn2001}, Gai and Kapadia \citeyearpar{gaik2010}, and Jackson and Pernoud \citeyearpar{jacksonp2020}. The special
case in which ${D_{ij}}={0}$ for all $i>0,j>0$ corresponds to Elliott,
Golub, and Jackson \citeyearpar{elliottgj2014}.

The actual market value of a bank is the value of the bank that accrues to final (i.e., outside) shareholders, which is $S_{0i} V_i$.
We show next that, given this definition of bank values, what eventually accrues to outside shareholders precisely equals the total value of primitive assets net of bankruptcy costs in the economy.
It is less cluttered to work the book values $V_i$s as it makes debts transparent, but a simple rescaling by $S_{0i}$ yields the corresponding market values.  Once rescaled, this is equivalent to the derivation of market values in Elliott,
Golub, and Jackson \citeyearpar{elliottgj2014}, for instance.\footnote{Here $S_{0i}$ is the same as $\widehat{C}_{ii}$ in their notation.}

\subsubsection{Consistency of Bank Values}\label{consistencyV}

From equation \ref{eq-bookvalue-bankruptcy} that bank book values solve
\[\mathbf{V}= \mathbf{q} \mathbf{p}  + \mathbf{d}^A(\mathbf{V}) + \mathbf{S}\mathbf{V}^+ - \mathbf{D}^L - \mathbf{b}(\mathbf{V},\mathbf{p}).\]
Written this way, the book value times the outside share percentage of a publicly held organization coincides with its market value.  We verify that in this section.

Indeed as argued by both Brioschi, Buzzacchi, and Colombo \citeyearpar{brioschibc1989} and  Fedenia, Hodder, and Triantis \citeyearpar{fedeniaht1994}, the ultimate (non-inflated) value of an organization to the economy -- what we call the ``market'' value --  is well-captured by the equity value of that organization that is held by its \emph{outside} investors -- or the {\sl final} shareholders who are private entities that have not issued shares in themselves. This value captures the flow of real assets that accrues to final investors of that organization.
We know argue that this is characterized by the above values, when weighted by the outside shares and accounting for debts and bankruptcy costs.

Summing the book values gives
\begin{align*}
\sum_{i\neq 0}V_i=\sum_{i\neq 0}\sum_k q_{ik} p_k+ \sum_{i\neq 0}  d_i^A(\mathbf{V})+ \sum_{i\neq 0}\sum_{j\neq 0} S_{ij} V_j  -\sum_{i\neq 0} D_i^L -b_i(\mathbf{V},\mathbf{p})
\end{align*}
It is helpful to decompose debtors $j$ into three categories: (i) those that are solvent and pay back in full $d_{ij}(\mathbf{V})=D_{ij}$, (ii) those that default but still have enough assets net of bankruptcy costs to make partial payments $d_{ij}(\mathbf{V}) >0$, and, (iii) those that default fully $d_{ij}(\mathbf{V})=0$. The latter are banks $j$ with $\mathbf{q}_j\mathbf{p}+d_j^A(\mathbf{V}) + \sum_k S_{jk}V_j^+ - b_j(\mathbf{V}, \mathbf{p})<0$, or equivalently, $V_j+D_j^L<0$.
The above expression can then be rewritten as
\begin{align*}
\sum_{i\neq 0}V_i =\sum_{i\neq 0}\sum_k q_{ik} p_k+ \sum_{i\neq 0}  \sum_{j\text{ solvent}}D_{ij}+ \sum_{i\neq 0}\sum_{\substack{j\text{ insolvent}\\V_j+D_j^L\geq 0}}\frac{D_{ij}}{D_j^L}[V_j+D_j^L]+ \sum_{j\text{ solvent}}(1-S_{0j}) V_j\\  -\sum_{i\neq 0} D_i^L -\sum_{i\neq 0}b_i(\mathbf{V},\mathbf{p}) \\[7pt]
\iff \sum_{i\neq 0}V_i =\sum_{i\neq 0}\sum_k q_{ik} p_k+ \sum_{i\neq 0}  \sum_{j\text{ solvent}}D_{ij}+\sum_{\substack{j\text{ insolvent}\\V_j+D_j^L\geq 0}}[V_j+D_j^L-d_{0j}(\mathbf{V})]+\sum_{j\text{ solvent}}(1-S_{0j}) V_j \\  -\sum_{i\neq 0} D_i^L -\sum_{i\neq 0}b_i(\mathbf{V},\mathbf{p}) \\[7pt]
\iff  d_0^A(\mathbf{V})-D_0^L+\sum_{j\text{ solvent}}S_{0j} V_j+\sum_{\substack{j\text{ insolvent}\\V_j+D_j^L< 0}}[V_j+D_j^L]=\sum_{i\neq 0}\sum_k q_{ik} p_k-\sum_{i\neq 0}b_i(\mathbf{V},\mathbf{p}).
\end{align*}
Recall that banks with $V_j+D_j^L<0$ are insolvent and whose bankruptcy costs exceed the value of their assets. These banks are not able to repay anything to their creditors.   Their excess failure costs are ultimately born by outside investors, and so we need to account for those too.  Thus, the overall value that eventually accrues to outside investors is $V_0 = d_0^A(\mathbf{V})-D_0^L+\sum_{j\text{ solvent}}S_{0j} V_j+ \sum_{\substack{j\text{ insolvent}\\V_j+D_j^L< 0}}[V_j+D_j^L]$.
It follows that
$$V_0 = \sum_{i\neq 0}\sum_k q_{ik} p_k-\sum_{i\neq 0}b_i(\mathbf{V},\mathbf{p}).$$
The total value accruing to all private investors---the market equity value plus net debts to outside investors---
equals the total value of primitive investments net of bankruptcy costs.

\subsubsection{Risky Portfolio Choices with Debt and Equity}

We now show how the results of Proposition \ref{risky} extend to the case with equity holdings.

Shareholders maximize $\mathbb{E}_\mathbf{p} \left[V_i(q_i, q_{-i};\mathbf{p})^+\right]$, that is they solve
\[
\max_{q_i\in [0,1]} \, \mathbb{E}_\mathbf{p} \left[ \left( q_i p_i + (1-q_i)(1+r)
+ \sum_{j\neq i} S_{ij} V_j(q_i, q_{-i};\mathbf{p})^+
+ d_i^{A}(\mathbf{V}(q_i,q_{-i};\mathbf{p})) -D_i^{L} \right)^+ \right].
\]
Here we do not allow for short sales (of either asset), which limits $q_i\in [0,1]$.  As will be clear below, the analysis extends to short sales: the bank would choose to short the risk-free asset and leverage its investment in the risky asset.

The value of $i$'s interbank equity and debt claims $ S_{ij}V_j(q_i, q_{-i};\mathbf{p})^+$  and $d_i^{A}(\mathbf{V}(q_i,q_{-i};\mathbf{p}))$ depends on banks' equilibrium values $\mathbf{V}$, and hence on the full vector of investment decisions $(q_i,q_{-i})$. Network interdependencies can generate market discipline, as a risky investment of Bank $i$ can trigger losses for some of its counterparties and feedback to itself.
We give sufficient conditions on the network structure for such market discipline not to arise endogenously.

We  say that $i$ is ``{\sl at risk of generating discontinuous feedback},''
if $i$ is part of a dependency cycle that involves \emph{both}
debt and equity, and that begins with a bank having an
equity claim on $i$. Formally, $\exists i_0, \dots i_K$ for some $K$
such that: $i_0=i_K=i$, $S_{i_1i_{0}}>0$, and $D_{i_{l}i_{l-1}}>0$ or $S_{i_{l}i_{l-1}}>0$
for each $l\geq 1$ with at least one $l$ such that $D_{i_{l}i_{l-1}}>0$.
Note that this is a condition on the network structure, which is a primitive of our model.

\begin{proposition}\label{risky2}
If Bank $i$ is not at risk of generating discontinuous feedback then it invests fully in the risky portfolio.
\end{proposition}

Generating discontinuous feedback means that, by making a safer investment,
Bank $i$ can prevent another bank's default, which triggers debt repayments that come back to benefit $i$.
Thus, this requires that $i$'s outcome can lead another bank to default while $i$ is still solvent, and
also that the other bank owes $i$ a sizeable sum.
Incentives to take (excessively) risky positions can be mitigated by
this sort of interdependencies.  However, that requires an extreme feedback
effect in which there {\sl must} be a nontrivial chance of driving a
counterparty in whom $i$ has a large stake into bankruptcy {\sl without} having $i$ become bankrupt.
In Section \ref{countervail}, we provide an example in which Lemma \ref{risky} does not apply, and discuss how
this example corresponds to the case in which a bank has the least incentives to take risks.

Proposition \ref{correlation1} extends directly to this more general model that allows for equity interdependencies. Furthermore, observe also that the inequality in Proposition \ref{correlation1} always holds weakly for a bank that is not at risk of discontinuous feedback. Indeed, for it to fail to hold, a Bank $i$ must gain more by having a high return when others have low returns. This can only be the case if such a high return for Bank $i$ can prevent others' defaults and feed back to $i$ in a way that strictly increases $i$'s net payoff. So $i$ would have to belong to a dependency cycle. Furthermore, that cycle must start with an equity claim on $i$, as otherwise $i$ cannot receive more from its counterparties than what it pays to them itself when $\mathbf{p}_{-i}=\mathbf{0}$.\footnote{Formally, for the inequality  in Proposition \ref{correlation1} not to hold, it must be that $V^+_i(p_i=R_i, \mathbf{p}_{-i}=\mathbf{R}_{-i}) -V^+_i(p_i=0, \mathbf{p}_{-i}=\mathbf{R}_{-i})< V^+_i(p_i=R_i, \mathbf{p}_{-i}=\mathbf{0})$. This can only be the case if both $V^+_i(p_i=R_i, \mathbf{p}_{-i}=\mathbf{0})>0$ and $V^+_i(p_i=0, \mathbf{p}_{-i}=\mathbf{R}_{-i})>0$, that is, if $i$ remains solvent when it has a low return but all others have high returns and when it is the only one with a high return. The condition is then
\begin{align*}
\sum_j S_{ij}[V_j^+(p_i=R_i, \mathbf{p}_{-i}=\mathbf{R}_{-i}) - V_j^+(p_i=0, \mathbf{p}_{-i}=\mathbf{R}_{-i}) ]+d_i^A(p_i=R_i, \mathbf{p}_{-i}=\mathbf{R}_{-i})-d_i^A(p_i=0, \mathbf{p}_{-i}=\mathbf{R}_{-i})\\
<\sum_j S_{ij}V_j^+(p_i=R_i, \mathbf{p}_{-i}=\mathbf{0})+d_i^A(p_i=R_i, \mathbf{p}_{-i}=\mathbf{0})-D_i^L.
\end{align*}
Recall that the left-hand side is weakly positive. If $i$ is not at risk of discontinuous feedback, then his portfolio return can \emph{only} impact the value of its assets on others by changing how much $i$ repays its counterparties. When $i$ is the only one with a high return, the total value of its debt/equity claims on others  $\sum_j S_{ij}V_j^+(p_i=R_i, \mathbf{p}_{-i}=\mathbf{0})+d_i^A(p_i=R_i, \mathbf{p}_{-i}=\mathbf{0})$ cannot be more than what it itself pays back to them $D_i^L$. Thus the RHS cannot be strictly positive, and the inequality cannot hold.
} For a bank not to want to correlate its portfolio, it must hence to be at risk of
discontinuous feedback. Financial networks in which banks want to correlate
their portfolios are then also those in which they want to take on as much risk
as possible (Lemma \ref{risky}), and vice versa.

\subsubsection{Centralities with Debt and Equity}

\[
NFC_i(\mathbf{q},\mathbf{q}_i';\mathbf{D},\mathbf{S})=\mathbb{E}_{\mathbf{p}} \left[
-\sum_j \left(b_j(\mathbf{V}(\mathbf{q}),\mathbf{p}) \right)- b_j (\mathbf{V}(\mathbf{q}_i', \mathbf{q}_{-i}),\mathbf{p}) \right],
\]
This is the total impact on the economy  that comes
from a change in $i$'s investment strategy, based on the bankruptcy costs that are incurred.
This is {\sl net} since it is the impact beyond the direct change in $i$'s portfolio value.
If there are no changes in bankruptcy costs, then the net financial centrality of $i$ is 0.

Another important concept is the impact of guaranteeing to bailout a particular bank.
Define a bank's {\sl bailout centrality} to be
\[
BC_i(\mathbf{q};\mathbf{D},\mathbf{S})= \mathbb{E}_{\mathbf{p}} \left[
\sum_j  b_j\left(\mathbf{V},\mathbf{p} \right) - b_j\left(V_i^+, \mathbf{V}_{-i},\mathbf{p}\right) \right],
\]
where the $\mathbf{V}_{-i}$ in $b_j\left(V_i^+, \mathbf{V}_{-i},\mathbf{p}\right)$ are calculated presuming that $i$ does not default on any payments and has value $V_i^+$,
while the $\mathbf{V}$ in $ b_j\left(\mathbf{V},\mathbf{p} \right)$ are not.
This is the total difference in overall bankruptcy costs if a bank is insured by the government and bailed out whenever it becomes insolvent compared to
a world in which it is left to fail.\footnote{Note that a bank is bailed out up to the point where it is made solvent, so the value of to its shareholders remains at zero (if $V_i<0$ then $V_i^+=0$).}

With equity, cascades do not follow direct default chains but can skip a bank --
that is,  Bank $k$ could own $j$ who owns $i$.  Even if $j$ does not default, its value could go down
if $i$ defaults, which could indirectly cause $k$ to default. Hence with equity we can have indirect
failures, whereas with debt there must only be chains of direct failures.

In a network that only involves debt, and where investment portfolios include a risk-free asset,  then  Bailout Centrality is equal to Net Financial Centrality where $q_i'$ is set to have enough of the risk free asset to guarantee $i$'s solvency.

\subsubsection{An Example with Countervailing Incentives}\label{countervail}

\
 There are situations in which debt combined with equity can offer a
 countervailing incentive that leads to less risky
 investments, if there is a special sort of discontinuous feedback.
 As suggested by Lemma \ref{risky}, such a feedback effect only arises if the bank is part of a cycle that involves \emph{both} equity claims and debt claims. We illustrate this effect in the following example.  Essentially, it has to be that a low portfolio return would not directly cause the bank to become insolvent, but would instead lead one of its equity-holding counterparties to become insolvent, which would then lead to a loss of a payment back to the original bank.
\begin{figure}[!h]
\begin{center}
\begin{tikzpicture}[scale=1]
\definecolor{chestnut}{rgb}{0.8, 0.36, 0.36};
\definecolor{afblue}{rgb}{0.36, 0.54, 0.66};
\foreach \Point/\PointLabel in {(0,0)/1, (3,0)/2 }
\draw[ fill=afblue!40] \Point circle (0.35) node {$\PointLabel$};
\draw[<-, thick] (0.4,0.4) to [out=40,in=135] node[midway,above]  {$D_{12}=d$} (2.6,0.4);
\draw[->, thick] (0.4,-0.4) to [out=-40,in=-135] node [ midway, below]  {$S_{21}=s$} (2.6,-0.4);
  \end{tikzpicture}
  \end{center}
  \captionsetup{singlelinecheck=off}
  \caption[]{\small A network in which Lemma \ref{risky} does not apply. Bank 2 owes $d$ to Bank 1,
  and  owns an equity share $s$ on Bank 1. Each bank can either invest in the safe asset, or in her own
  risky asset that pays a positive return $R$ with probability $\theta$. Suppose $s(1+r)\geq d$
  so that Bank 1 can prevent Bank 2's default if it has a portfolio realization above the risk-free rate.
  We want to show that for some parameter values, both bank fully investing in the risky asset---$q_1=q_2=1$---is not an equilibrium. Suppose $q_2=1$. Choosing $q_1=1$ yields an expected payoff to Bank 1 of
  \begin{align*}\mathbb{E}[V_1(q_1=1,q_2=1)] = \theta^2(R+d)+\theta(1-\theta)(R+d)+\theta(1-\theta)d = \theta[R+(2-\theta)d].\end{align*}
  However, by choosing a safer portfolio, Bank 1 could prevent 2's default whenever none of the risky assets pays off. This requires choosing $q_1$ such that
  \[ sV_1(q_1, q_2=1; p_2=0) = s[(1-q_1^*)(1+r)+d] \geq d.\]
  The largest $q_1$ satisfying this condition solves $(1-q_1^*)(1+r) = \frac{(1-s)}{s}d$. Bank 1's expected payoff when choosing such portfolio is then
 \begin{align*}
 \mathbb{E}[V_1(q_1=q_1^*,q_2=1)] = (1-q_1^*)(1+r) + d + \theta q_1^* R = \theta R + d -\frac{d(1-s)}{s(1+r)}[\theta R-(1+r)].\end{align*}
 Such safer portfolio is better than fully investing in the risky asset for Bank 1 as soon as
   \begin{align*}\mathbb{E}[V_1(q_1=q_1^*,q_2=1)] >\mathbb{E}[V_1(q_1=1,q_2=1)] \iff 1 -\theta(2-\theta)>\frac{(1-s)}{s} \frac{1}{1+r}[\theta R - (1+r)].
   \end{align*}
If the risk premium is not too high, it is optimal for Bank 1 to choose a safer portfolio
so as to prevent its debtor's default, and get its payment of $d$ back with certainty.
This is however only possible when Bank 2 has some equity claim on 1.
  }\label{counter_example}
\end{figure}

This example illustrates some of the nuances of financial interdependencies.
Which mix of debt and equity is best for incentives depends on multiple features.
Both debt and equity generally incentivize banks to fully invest in the risky portfolio.
For them to choose safer investments, it is necessary to have a mix of debt and equity that
allows for discontinuous indirect effects of one's return on its own value through the
network.\footnote{The network in Figure \ref{counter_example} may seem odd as a bank is a
creditor of its own shareholder. However, for instance, this was one of the main characteristics
of the Icelandic
financial system right before its collapse in 2008: the three largest banks had lent funds to
clients so that they could buy the bank's own shares (see the report of the Special Investigation Commission \citeyearpar{SICreport}). }

The network depicted in Figure \ref{counter_example} has the
largest potential for such a feedback effect. Indeed, this sort of
feedback  starts with some Bank $j$ -- here Bank 2 -- having
some (direct or indirect) equity claim on another $i$ -- here Bank 1.
If this claim is large, then the latter bank need not invest a lot
in the safe asset to prevent the former bank's default. Hence the
larger this claim, the lower the cost of preventing someone's default.
Note that an indirect equity claim can always be replicated by a direct
claim that is at least as big; hence a single bank having a large direct
equity claim on another, as it is the case in the example, is most likely
to generate countervailing incentives.  This is, however, the case only if
the losses induced by the bank's default can feedback to $i$. This is true
for instance if $j$ owes some debt to $i$, or if $i$ has a claim on one of
$j$'s creditors. Again, the largest feedback can always be induced by $i$
having a direct, high enough, debt claim on $j$.

Figure \ref{counter_example} then represents the network in which a bank
has the least incentives to invest in the risky asset.
Thus, it provides
a sufficient condition for banks to fully invest in the risky asset
even when at risk of discontinuous feedback.  If no single bank
has a total (direct or indirect) equity claim of more than
$\frac{\theta R-(1+r)}{(1-\theta)^2(1+r)}$ on another, then investing
fully in the risky asset is the only equilibrium outcome.

\subsubsection{More Complicated Interbank Contracts Beyond Debt and Equity}

We now discuss bank values when contracts are not restricted to debt and equity.
A general form of contract between organizations $i$ and $j$ is denoted by $f_{ij}(\mathbf{V},\mathbf{p})$ and can depend on the value of organization $j$ as well as the value of other organizations. This represents some stream of payments that $j$ owes to $i$, usually in exchange for some good, payments, or investment that has been given or promised from $i$ to $j$.

The value $V_i$  of an organization $i$
is then
\begin{equation}
V_i=\sum_k q_{ik} p_k+ \sum_j  f_{ij}(\mathbf{V},\mathbf{p}) - \left[ \sum_j  f_{ji}(\mathbf{V},\mathbf{p})- S_{ji}(\mathbf{V})V_i^+\right] - b_i(\mathbf{V},\mathbf{p}), \label{eqn:V1}
\end{equation}
where $ f_{ji}(\mathbf{V},\mathbf{p})- S_{ji}(\mathbf{V})V_i^+$ accounts for the fact that
debt and contracts other than equity are included as liabilities in a book value calculation.\footnote{This more general model also embeds that of Barruca et al. \citeyearpar{barucca2016} in which banks hold debt on each other, but these debt claims are not valued under full information: they allow for uncertainty regarding banks' external assets and ability to honor their interbank liabilities, whose face value may then be discounted depending on available information.   Financial contracts as defined here can capture this kind of uncertainty if $f_{ij}$ equals the expected payment from $j$ to $i$ given some information---e.g. a subset of known bank values or primitive asset values.}

Under monotonicity assumptions on financial contracts, there exist consistent values for
banks by Tarski's fixed point theorem.
That is, there exists a fixed point to the above system of equations.
This is true whenever $b_i(\mathbf{V},\mathbf{p})$ is nonincreasing in $\mathbf{V}$ and
bounded (supposing that the costs cannot exceed some total level),  each $f_{ij}(\mathbf{V},\mathbf{p})$
is also nondecreasing in $\mathbf{V}$, $\sum_j  f_{ji}(\mathbf{V},\mathbf{p})- S_{ji}(\mathbf{V})V_i$
is nonincreasing in $\mathbf{V}$, and either $f$ is bounded or possible values of $\mathbf{V}$ are bounded.
Moreover, from Tarski's theorem,
it also follows that the set of equilibrium bank values forms a complete lattice. Discontinuities, which come from bankruptcy costs and
potentially the financial contracts themselves, can thus lead to multiple solutions for organizations' values.

\begin{figure}[!h]
\centering
\subfloat[Non-Decreasing Financial Contracts]{\includegraphics[width=0.48\textwidth]{contracts.tikz}}\quad
\subfloat[Non-Increasing Financial Contracts]{\includegraphics[width=0.48\textwidth]{contracts2.tikz}}
\end{figure}

When financial contracts are not increasing functions of $\mathbf{V}$,
there may not exist an equilibrium solution for bank values.
For instance, as soon as some banks insure themselves against
the default of a counterparty or bet on the failure of another, simple accounting rules may not yield consistent values for all organizations in the financial network. We illustrate this in the following example.

\paragraph{Example of Non-Existence of a Solution for $\mathbf{V}$: Credit Default Swaps.}
Consider a financial network composed of $n=3$ organizations, each of which having its own
proprietary asset. For simplicity all assets have the same value $p_i=2$ for $i=1,2,3$.
The values of organizations
are linked to each other through the following financial contracts: organization 2 holds debt from 1 with face value $D_{21}=1$; 2 is fully insured against 1's default through a CDS with organization $3$ in exchange of payment $r=0.4$; finally 1 holds a contract with 3 that is linearly decreasing in 3's value. Suppose an organization defaults if and only if its book value falls below its interbank liabilities, in which case it incurs a cost $ \beta = 0.1$. Formally, the contracts are
\begin{align*}
f_{21}(\mathbf{V}) = D_{21}\mathbbm{1}_{V_1\geq  0}\\
f_{23}(\mathbf{V}) = D_{21}\mathbbm{1}_{V_1<  0}\\
f_{32}(\mathbf{V}) = r\mathbbm{1}_{V_1\geq  0}\\
f_{13}(\mathbf{V}) = -0.5V_3.
\end{align*}
Note that organization 2 and 3 never default: the former's value is always at least $2-r>0$ and
the latter's is at least $2- D_{21}>0$. We then check that there is no solution in which organization 1
is solvent. In such a case, $V_3 = 2+ r$ and $V_1 = 2 -0.5 V_3 - D_{21} = -0.2<0$: but then Bank 1 defaults, which
is a contradiction. Finally suppose that 1 defaults.
Then $V_3 = 2- D_{21}$ and $V_1 = 2-0.5V_3-\beta = 1.4>0$, another contradiction.

The above model of financial networks incorporates both
debt and equity contracts, and thus generalizes existing models used to understand systemic
risk that are either built with debt interdependencies
(e.g., Eisenberg and Noe
\citeyearpar{eisenbergn2001}, Gai and Kapadia \citeyearpar{gaik2010}, and Jackson and Pernoud \citeyearpar{jacksonp2020})
or with equity-like interdependencies (e.g., Elliott, Golub, and Jackson \citeyearpar{elliottgj2014}).
Our model is also related to those used in the
accounting and asset valuation literatures (Suzuki \citeyearpar{suzuki2002},
Fischer \citeyearpar{fischer2014}) to derive equity prices when firms also
issue debt. We generalize those models
to include multiple types of contracts, while also allowing for bankruptcies and
discontinuous costs upon insolvency, since those play a key role
in inefficiencies and externalities in financial networks.
They furthermore  provide a lens into many other contracts, as many swaps and derivatives
involve either fixed payments or payments that depend on the realization of the
value of some investment of one of the parties, like a combination of debt and equity.

\subsection{Investment Incentives under Risk Aversion}\label{riskaversion}

We briefly investigate banks' investment incentives under risk aversion, and show why risk aversion does
not resolve the inefficiency. Consider the same investment problem as in Section \ref{overly},
except that banks are now risk averse. Let $u$ be their Bernoulli utility function,
which is assumed to be common to all banks and all outside investors.

We here restrict attention to financial networks composed of debt only. Equity holdings make the
analysis significantly more complex as they allow for some risk sharing.\footnote{For instance,
consider a setting in which two banks have equity claims on each other, and their risky
investments are negatively correlated. Then investing more in the risky asset for $i$ allows $j$ to
diversify its portfolio more. Furthermore, a high portfolio realization for $i$ may matter more for $j$
than $i$---if, e.g. the latter has other valuable assets, while the former does not --
making the problem far from trivial.} Then socially efficient investments are much more complicated,
as they no longer simply balance higher returns and larger expected bankruptcy costs, but also account for risk sharing.

Given others' investment decisions $q_{-i}$, Bank $i$ solves
\[
\max_{q_i\in [0,1]} \,  \mathbb{E}\left( u \left[ \left( q_i p_i + (1-q_i)(1+r)
+ d_i^{A}(\mathbf{p},q_i,q_{-i}) -D_i^{L} \right)^+ \right]\right).
\]

The socially efficient investment solves
\begin{align*}
\max_{q_i\in [0,1]} \,  \mathbb{E}&\left(u \left[ \left( q_i p_i + (1-q_i)(1+r)
+ d_i^{A}(\mathbf{p},q_i,q_{-i}) -D_i^{L} \right)^+ \right]\right)\\[5pt]
&+\mathbb{E}\left(\sum_{j\neq i} u \left[ \left( q_j p_j + (1-q_j)(1+r)
+ d_j^{A}(\mathbf{p},q_i,q_{-i}) -D_j^{L} \right)^+ \right]\right).
\end{align*}
It is clear that the additional term in the planner's objective can only be decreasing in $q_i$. Indeed, Bank $i$'s investment only enters Bank $j$'s utility through the debt repayment of $i$ to $j$. Since a riskier investment can only increase the probability that $i$ defaults on its payment to $j$ in a network composed of only debt, $D_j^{A}(\mathbf{p},q_i,q_{-i})$ is weakly decreasing in $q_i$ for all $q_{-i}$, all $\mathbf{p}$. Hence risk aversion does not restore the efficiency of banks' investment decisions: they still take on too much risk.

\end{appendices}
\end{spacing}
\end{document}